\documentclass[lettersize,journal]{IEEEtran}
\usepackage{amsmath,amssymb,amsfonts,amstext}
\allowdisplaybreaks[4]
\usepackage{amsthm} 
\usepackage{algorithmic}
\usepackage{algorithm}
\usepackage{enumitem}
\usepackage{array}
\usepackage{color}
\usepackage[caption=false,font=normalsize,labelfont=sf,textfont=sf]{subfig}
\usepackage{textcomp}
\usepackage{stfloats}
\usepackage{url}
\usepackage{verbatim}
\usepackage{graphicx}
\usepackage{cite}
\usepackage{multirow}
\newtheorem{Definition}{Definition}

\newtheorem{Lemma}{Lemma}
\newtheorem{assumption}{Assumption}
\newtheorem{Theorem}{Theorem}
\theoremstyle{assumption}
\newtheorem{remark}{Remark}
\theoremstyle{assumption}
\newcommand{\black}[1]{\textcolor{black}{#1}}

\begin{document}

\title{Event-triggered control and communication for single-master multi-slave teleoperation systems with Try-Once-Discard protocol}

 \author{Yuling~Li,~\IEEEmembership{Member,~IEEE,}
	Chenxi~Li,
	Kun~Liu,~\IEEEmembership{Senior  Member,~IEEE,}
	Jie~Dong, 
	Rolf Johansson,~\IEEEmembership{Fellow,~IEEE}
	
	\thanks{Y. Li, C. Li, and J. Dong  are with the School of Automation and Electrical Engineering,
		University of Science and Technology Beijing, Beijing, 100083, P.~R.~China, and the Key Laboratory of Knowledge Automation for Industrial Processes, Ministry of Education, Beijing 100083, P.~R.~China. \protect E-mails: yuling@ustb.edu.cn, \protect lchenxi380223050@163.com, dongjie@ies.ustb.edu.cn. K. Liu is with the School of Automation, Beijing Institute of Technology, Beijing 100081, P.~R.~China.  \protect E-mail: kunliubit@bit.edu.cn (Corresponding author).
R. Johansson is with the Dept. Automatic Control, Lund University, P.O. Box 118, 22100 Lund, Sweden. He is member of the eLLIIT Excellence Center at Lund University. \protect E-mail: Rolf.Johansson@control.lth.se.}
\thanks{This work was jointly supported by the National Key Research and Development Program of China (No. 2023YFB4706900), the National Natural Science Foundation of China (No. 62573041, 62273041, U24A20264), the Open Projects of the Institute of Systems Science, the Beijing Natural Science Foundation (No. L252123), Beijing Wuzi University (BWUISS11, BWUISS44).}
}


\markboth{IEEE Transactions on Cybernetics}%
{Shell \MakeLowercase{\textit{et al.}}: Event-triggered control and communication for SMMS teleoperation systems with Try-Once-Discard protocol}


\maketitle

\begin{abstract}
Single-master multi-slave (SMMS) teleoperation systems can perform multiple tasks remotely in a shorter time, cover large-scale areas, and adapt more easily to single-point failures, thereby effectively encompassing a broader range of applications. As the number of slave manipulators sharing a communication network increases, the limitation of communication bandwidth becomes critical. To alleviate bandwidth usage, the Try-Once-Discard (TOD) scheduling protocol and event-triggered mechanisms are often employed separately. In this paper, we combine both strategies to optimize network bandwidth and energy consumption for SMMS teleoperation systems.
Specifically, we propose event-triggered control and communication schemes for a class of SMMS teleoperation systems using the TOD scheduling protocol. Considering dynamic uncertainties, the unavailability of relative velocities, and time-varying delays, we develop adaptive controllers with virtual observers based on event-triggered schemes to achieve master-slave synchronization. Stability criteria for the SMMS teleoperation systems under these event-triggered control and communication schemes are established, demonstrating that Zeno behavior is excluded. Finally, experiments are conducted to validate the effectiveness of the proposed algorithms.

\end{abstract}


\normalsize
\section{Introduction}

\IEEEPARstart{T}{eleoperation}  systems are a kind of  robotic systems that transmit commands from a local manipulator to a remote manipulator, enabling the execution of desired tasks at a distance. These systems have been extensively applied in fields such as space and underwater exploration, telesurgery, and so on~\cite{Shahbazi2018}, \cite{yan2024survey,black2024TMRB}. As the complexity of remote tasks increases, relying on a single slave manipulator within teleoperation systems becomes increasingly challenging. Consequently, single-master multi-slave (SMMS) teleoperation systems (illustrated in Fig.\ref{fig
}) have gained significant attention \cite{Lu2023TASE,Shahbazi2018}.

\begin{figure}[!htb]
\centering
\includegraphics[width=0.9\linewidth]{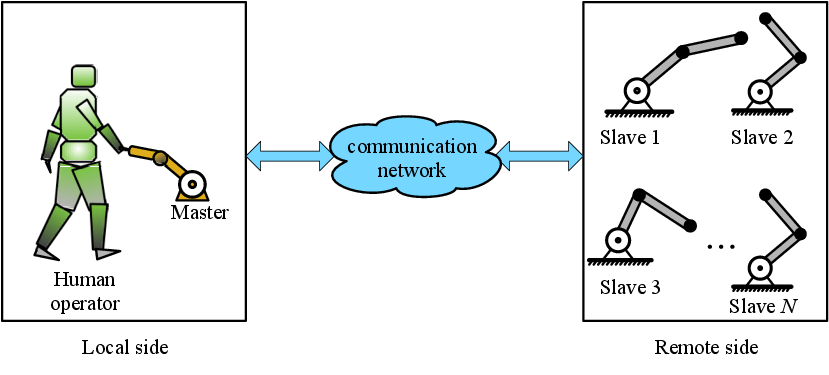}
\caption{The SMMS teleoperation system}
\label{fig
}
\end{figure}

As the number of slave manipulators sharing a communication network increases in SMMS teleoperation systems, the limitation of communication bandwidth becomes a critical concern. To enhance bandwidth utilization, various scheduling protocols have been introduced in teleoperation systems~\cite{li2019bilateral}. These protocols generally fall into three categories: Round-Robin (RR), Try-Once-Discard (TOD), and stochastic scheduling protocols~\cite{zou2017ultimate, cheng2021output}. The position synchronization problem for SMMS teleoperation systems under RR and TOD scheduling protocols has been thoroughly studied in~\cite{li2021distributed,li2019bilateral}.

To further reduce energy  consumption and  network workload, event-triggered mechanisms have also been  used in many networked systems~\cite{Peng2018survey,Dolk2024TAC}, such as static neural networks \cite{vadivel2023event}, fuzzy systems \cite{dong2020dynamic, Yang2023Tcy} non-homogeneous Markov
	switching systems \cite{chen2022an},  linear multi-agent systems \cite{liu2024Tcyber, Jin2024Tcyber}, networked Euler-Lagrange systems \cite{wei2021velocity, Ngo2021TCNS, Hao2024Tcyber}, and so on. These mechanisms can be applied to controller-to-actuator channels (known as event-triggered control mechanisms) \cite{Xu2024TASE} and sensor-to-network channels (known as event-triggered communication mechanisms) \cite{Hao2024Tcyber, liu2024Tcyber, Jin2024Tcyber}.

In event-triggered control strategies, control signals are updated only when a specific triggering condition is met. Similarly, in event-triggered communication strategies, current measurements are transmitted to the communication channel only when a triggering event occurs. Recently, event-triggered control and communication strategies have been widely explored in teleoperation systems~\cite{wang2021event, zhao2022adaptive, Xu2024MPC_TASE, Li2024AJC, Fu2023CEP, gu2024event}. For instance, an event-triggered prescribed-time fuzzy control strategy was developed to ensure the prescribed-time stability of the closed-loop system~\cite{wang2021event}. A fixed-time event-triggered control scheme was formulated for bilateral teleoperation systems, accounting for uncertain disturbances and asymmetric communication delays~\cite{zhao2022adaptive}. A novel self-triggered model predictive control framework, incorporating a high-order estimation mechanism, was proposed for the teleoperation of networked mobile robotic systems~\cite{Xu2024MPC_TASE}. An event-triggered finite-time variable gain active disturbance rejection control scheme was proposed for teleoperation parallel manipulators in~\cite{gu2024event}. In~\cite{Fu2023CEP}, event-triggered communication among multiple slave mobile robots was considered in SMMS teleoperation systems, and a predefined-time cooperative control scheme was proposed to synchronize the task-space position between the master and the slaves.

Limited communication bandwidth can lead to data collisions, particularly when multiple nodes compete for network access simultaneously. To address this issue, the simultaneous consideration of scheduling protocols and event-triggered mechanisms is essential. For example, in~\cite{dong2020dynamic, Yang2023Tcy}, event-triggering-based TOD was used to orchestrate data transmission for fuzzy systems, with model predictive controllers and sliding-mode controllers addressed, respectively. For nonhomogeneous Markov switching systems, an event-triggered RR protocol was proposed to reduce network workload in~\cite{chen2022an}. A dynamic event-based TOD protocol was proposed in~\cite{Cheng2022Tcy} for networked nonlinear systems. However, it should be noted that~\cite{dong2020dynamic, Yang2023Tcy, Cheng2022Tcy} did not specifically target teleoperation systems. For SMMS teleoperation systems, Li et al. also considered the event-triggered TOD protocol to improve network bandwidth utilization, but  assumed exact knowledge of the system dynamics~\cite{Li2024AJC}.

In practical applications, the dynamic models of manipulators are often not precisely known in advance, leading to dynamic uncertainties that can adversely affect system performance and even cause instability. To address these dynamic uncertainties, various advanced control techniques, including adaptive control, neural networks, and fuzzy logic, have been proposed in the literature. For example, a robust adaptive control algorithm for nonlinear teleoperation systems, leveraging the property of linearly parameterizable dynamics, was introduced in~\cite{kebria2020robust}. A globally stable adaptive fuzzy backstepping control scheme for bilateral teleoperation manipulators was developed in~\cite{chen2020adaptive}. The use of neural-network approximators for the control of uncertain teleoperation systems was explored in~\cite{li2021distributed}. However, these methods typically rely on relative velocities for control design, necessitating the real-time transmission of velocity information across the network. To overcome this challenge, a robust sliding control algorithm \textcolor{black}{that does not depend on relative velocities was developed for bilateral teleoperators in~\cite{liu2020on}}. Nonetheless,  the approach in \cite{liu2020on} does not  fully account for the constraints imposed by limited communication bandwidth and computational resources, presenting significant challenges for their implementation in practical teleoperation systems.

Motivated by the above observations, in this paper, we address both the event-triggered control and event-triggered communication problems for a class of SMMS teleoperation systems with the TOD scheduling protocol subject to dynamic uncertainties. It is noteworthy that the TOD protocol is superior in resource allocation compared to RR and stochastic scheduling protocols, as it fully considers the current system information during the scheduling process~\cite{dong2020dynamic}. However, the use of the TOD protocol introduces challenges in controller design due to the dynamic nature of data transmission order. Therefore, developing a model that simultaneously incorporates event-triggered mechanisms and the TOD scheduling protocol, along with devising control strategies that ensure the desired performance of SMMS teleoperation systems, is a significant challenge.
	 
	 The key contributions of this paper can be summarized as follows:

\begin{itemize}
	\item We propose two unified control frameworks for SMMS teleoperation systems that jointly address the challenges of TOD scheduling, event-triggered mechanisms, time-varying delays, and dynamic uncertainties. This comprehensive integration is rarely explored in prior work and addresses practical challenges in large-scale teleoperation systems.
	
	\item We design adaptive controllers based on virtual reference systems, which avoid differentiating discontinuous signals induced by TOD scheduling and event-triggering. This allows for effective parameter adaptation under uncertain dynamics without requiring real-time transmission of velocity signals, thereby overcoming limitations of existing adaptive and robust control approaches.
	
	\item We develop a novel stability analysis framework using discontinuous Lyapunov-Krasovskii functionals. This framework rigorously guarantees master-slave synchronization, ensures Zeno-free behavior of event-triggering mechanisms, and handles the discontinuities introduced by TOD scheduling and event-triggered communication.
	
\end{itemize}

\textit{Notations.} Let $T$ denote matrix transposition. $\mathbb{R}^n$ is the $n$-dimensional Euclidean space with vector norm $|\cdot|$, $\mathbb{R}^{n\times m}$ represents the set of $n\times m$  real matrices, $\mathbb{N}$ stands for the set of non-negative integers while $\mathbb{N}^+$ is the set of positive integers, respectively. $I$ is the identity matrix, and $*$ indicates symmetric entries in block matrices. A real symmetric matrix $P$ satisfies $P>0$ ($P<0$) if it is positive (negative) definite.
For any function $f: [0, \infty) \to \mathbb{R}^n$, define the $\mathcal{L}_\infty$-norm as $\|f\|_\infty := \sup_{t \ge 0} |f(t)|$, and the square of the $\mathcal{L}_2$-norm as $\|f\|_2^2 := \int_0^\infty |f(t)|^2 dt$. The $\mathcal{L}_\infty$ and $\mathcal{L}_2$ spaces consist of functions with finite $\|f\|_\infty$ and $\|f\|_2$, respectively. Time $t$ may be omitted when clear from context.

\section{Preliminaries and problem formulation}\label{sec:problem_formulation}

\subsection{Dynamic models}
The Euler-Lagrange equations of motion for the SMMS teleoperation system are given as follows \cite{kelly2005control,spong2020robot}: 
\begin{align}
&M_{m}(q_{m}) \ddot{q}_{m}+C_{m}(q_{m}, \dot{q}_{m})\dot{q}_{m}+G_{m}(q_{m})=f_{m}+\tau_{m}\label{eq:master}\\
&M_{s i}\left(q_{s i}\right) \ddot{q}_{s i}+C_{s i}\left(q_{s i}, \dot{q}_{s i}\right) \dot{q}_{s i}+G_{s i}\left(q_{s i}\right)=f_{s i}+\tau_{s i}\label{eq:slave}
\end{align}
where the subscripts $m$ and $si$ denote the master and the $i$th slave robot manipulator with $i=1,\dots, N$, respectively. For $z=m, s1, ....sN$,  $q_z, \dot{q}_z, \ddot{q}_z \in \mathbb{R}^{n}$  are the joint positions, velocities, acceleration vectors of the manipulators, respectively,  $M_z\in \mathbb{R}^{n \times n}$ is the inertia matrix, $C_{z}\left(q_{z}, \dot{q}_{z}\right)\in \mathbb{R}^{n \times n}$ represents the centripetal and coriolis torque, $G_z(q_z) \in \mathbb{R}^{n}$ embodies the gravity vector, $\tau_{z}\in \mathbb{R}^{n}$ is the applied control torque, $f_{z}\in\mathbb{R}^n$ is the external torque generated by the operator and the compliant environment interaction.
%
%

To facilitate the theoretical analysis in Section~\ref{sec:Controller Design and Stability Analysis}, the well-known properties of the robotic systems (\ref{eq:master})-(\ref{eq:slave}) with revolute joints are given as follows ( $z=m, s1, s2,..., sN$) \cite{kelly2005control,spong2020robot}:
\begin{enumerate}
	\item[P1]\cite{nuno2010an, wang2020differential}The inertia matrix $M_{z}\left(q_{z}\right)$ is positive-definite and there exist positive scalars $\lambda_{z}^{m}$, $\lambda_{z}^{M}$ such that $0<\lambda_{z}^{m} I \leq M_{z}\left(q_{z}\right) \leq \lambda_{z}^{M} I<\infty$.
	\item[P2]\cite{Lu2023TASE,nuno2010an, wei2021velocity, wang2020differential} $\forall \eta,q_{z}\in \mathbb{R}^{n}$, $\eta^{T}(\dot{M}_{z}(q_{z})-2 C_{z}(q_{z}, \dot{q}_{z}))\eta=0$.
	\item[P3]\cite{wang2021event} $\forall x,y,q_{z}\in \mathbb{R}^{n}$, there exists a positive
	scalar $c_{z}$ such that $\left|C_{z}(q_{z}, x) y\right| \leq c_{z}|x||y|$.
	\item[P4]\cite{nuno2010an} If $\ddot{q}_{z}$ and $\dot{q}_{z}$ are bounded, the time derivative of $C_{z}\left(q_{z}, \dot{q}_{z}\right)$ is bounded.
	\item[P5] \cite{wei2021velocity} The gravity force $G_z(q_z)$ is  bounded, i.e., there exists a positive constant $g_z$ such that $|G_z(q_z)|\leq g_z$.
	\item[P6]\cite{nuno2010an,wang2020differential} \label{property:P6} The dynamics in (\ref{eq:master}) and (\ref{eq:slave}) are linearly parameterizable as $M_z(q_z)x+C_z(q_z, \dot{q}_z)y-G_z(q_z)=Y_z(q_z, \dot{q}_z, x, y)\theta_z$, where $x, y \in \mathbb{R}^{n}$, $Y_z(q_z, \dot{q}_z, x, y)$ is the regressor and $\theta_z$ is a vector of unknown but constant parameters. 
\end{enumerate}

\begin{assumption}\label{amp:Y}
	For $z = m, s_1, s_2, ..., sN$, each element of $M_z, C_z, G_z, x, y$ is absolutely continuous and globally Lipschitz, which implies that $Y_z(q_z, \dot{q}_z, x, y)$ is absolutely continuous and globally Lipschitz. Hence, $Y_z(q_z, \dot{q}_z, x, y)$ and $\dot{Y}_z(q_z, \dot{q}_z, x, y)$ are bounded.
\end{assumption}

\subsection{Description of data transmission under communication constraints}


Generally, data are exchanged between the master and the slaves through a shared communication network. Due to the constraints of communication bandwidth, many unfavorable phenomena, such as network congestion and data collisions, cannot be avoided when large amounts of data are released simultaneously. To overcome these obstacles, a dynamic scheduling protocol, specifically the TOD protocol defined in Definition 1 below, is employed to orchestrate data transmission so that only one manipulator's data can be transmitted at a time. Since there is only one manipulator on the master side, the TOD protocol is applied only in the backward communication channel, i.e., the channel from the slaves to the master. By utilizing the TOD protocol, the transmission instants of the slaves, denoted as $s_0^s, s_1^s, \ldots, s_{k}^s, \ldots$ ($k \in \mathbb{N}$), at which information from one of the slaves is sent to the master, are generated. 
Let $x_{si}$ represent the data of the $i$th slave to be transmitted through the communication network. Hence, $x_{si}$ is sampled at specific instants, and the latest data transmitted to the master is characterized by  $\hat{x}_{si}$ using the TOD protocol, where only one manipulator is granted access to the communication network, while the information from the other slaves is held by zero-order holders (ZOHs),
  that is, 
\begin{eqnarray}\label{eq:updating_x_sk}
	\hat{x}_{s i}\left(s_{k}^s\right)=\left\{\begin{array}{c}{x_{s i}\left(s_{k}^s\right), i=i_{k}^{*}} \\ {\hat{x}_{s i}\left(s_{k-1}^s\right), i \neq i_{k}^{*}}\end{array}\right.
\end{eqnarray}
where ${i}_{k}^{*}$ is the active slave that obtains the access to the communication network. 
\begin{Definition} \label{def:TOD} 
	Let  $Q_{i}>0(i=1, \ldots, N)$ be the weighting matrices. At the transmission instant $s_{k}^s$, the TOD protocol is a protocol for which the active slave robot with the index $i_{k}^{*}$ is defined as any index that satisfies
	\begin{equation}\label{eq:weighting_matrices}
		\begin{aligned} \left|\sqrt{Q_{i_{k}^{*}}} \eta_{i_{k}^{*}}\right|^{2} &\geq \left|\sqrt{Q}_{i} \eta_{i}\right|^{2} , i= 1, \ldots, N\end{aligned}
	\end{equation}
	where 
	\begin{equation}\label{eq:greatest_error}
		\begin{aligned} 
			\eta_{i} &=\hat{x}_{si}\left(s_{k-1}^s\right)-x_{si}\left(s_{k}^s\right) \\ t & \in\left[s_{k}^s, s_{k+1}^s\right), \quad k \in \mathbb{N}, \hat{x}_{si}\left(\textcolor{black}{s_{-1}^s}\right)=0 \end{aligned}
	\end{equation}
	is the weighted transmission error.
\end{Definition}

\begin{assumption}\label{amp:sampling}
	The transmission intervals  are bounded, i.e, there exists a positive scalar $h$ such that
$
s_{k+1}^s-s_{k}^s\leq h$ $(k\in\mathbb{N})
$ where $h$ represents the maximum allowable transmission interval.
\end{assumption}
In this paper, we assume that the forward and the backward communication delays $T_m, T_{s}$ are unknown, asymmetric and time-varying. The assumption regarding the time delays is stated as follows:


\begin{assumption}
For $j=m, s$, there exist positive constants $d_j, p_j$ such that the unknown communication delays $T_j$ satisfy $0\leq T_j\leq d_j$ and $\dot{T}_j\leq p_j<1$. 
\label{amp:delay}
\end{assumption}

\subsection{Problem formulation}

For the teleoperation system with multiple slaves (\ref{eq:master})-(\ref{eq:slave}),  the coordinated motion between multiple slaves should be guaranteed by formation schemes, in which  the $i$th slave maintains a distance and orientation $\gamma_{i}$ from the formation's geometric center $\bar{q}_{s}:=\frac{1}{N} \sum_{i=1}^{N} q_{s i}$, where $\gamma_i\in\mathbb{R}^n$ is a constant vector, $\gamma_{i} \neq \gamma_{j}, \forall i \neq j, \sum_{i=1}^{N} \gamma_{i}=0$. Hence, this paper aims to achieve master-slave synchronization $\lim_{t\rightarrow\infty}\left(q_m-\bar{q}_s\right)=0$ for the teleoperation system (\ref{eq:master})-(\ref{eq:slave}) with time-varying delays and the TOD scheduling protocol by proposing an event-triggered control scheme and a control scheme under  event-triggered communication. 

\black{Specifically, this paper addresses control design problems under two different scenarios. } In the first case, the data to be transmitted for each manipulator are sampled by time-based samplers, and the sampled data are sent to the remote sides through the delayed communication network with the TOD protocol. The event-triggered control design is then addressed. In the second case, the manipulators' data are sampled by event-based samplers and then sent to the remote sides through the communication network with the TOD protocol, and the control design for the teleoperation system under the event-triggered communication is addressed. In the following section, these two issues are respectively addressed. 
\section{Main results} \label{sec:Controller Design and Stability Analysis}


\subsection{Event-triggered control}
In this section,  the event-triggered control design for the SMMS teleoperation system (\ref{eq:master})-(\ref{eq:slave})  is considered. 
The control framework is depicted in Fig.~\ref{fig:eventcontrolscheduling}.
\begin{figure}[tbh!]
	\centering
	\includegraphics[width=0.95\linewidth]{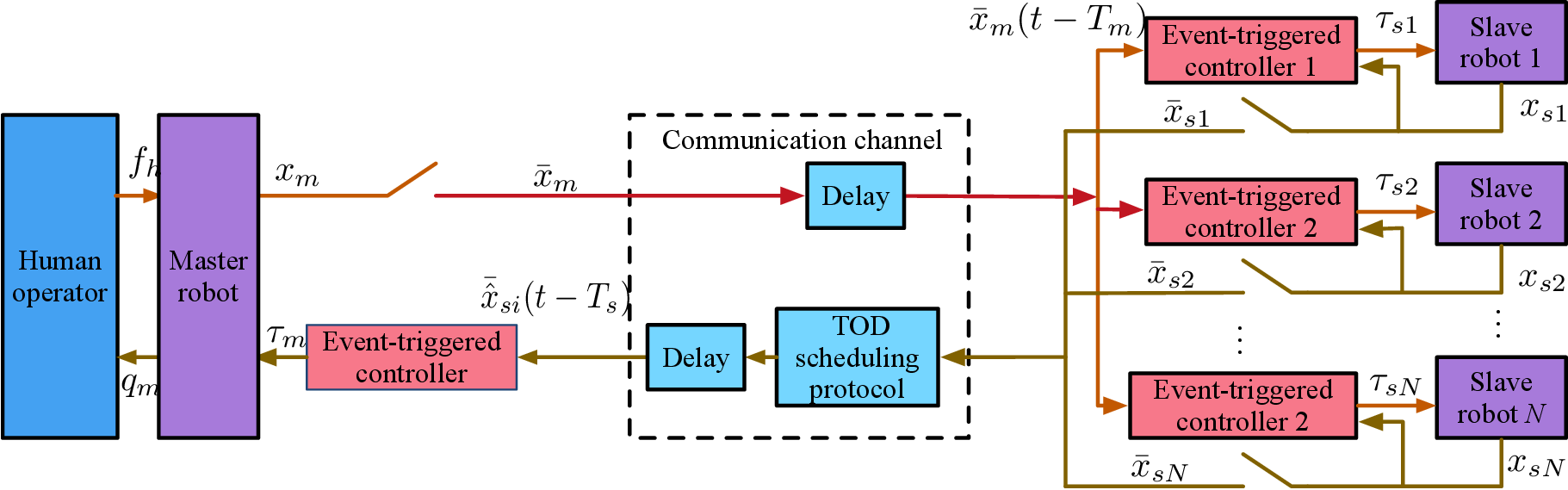}
	\caption{The event-triggered control framework for the SMMS teleoperation system with the TOD scheduling protocol.}
	\label{fig:eventcontrolscheduling}
\end{figure}

For simplicity, we assume that the data at the master and the slave sides are sampled and transmitted over the network synchronously at the instants  $s_{k}^s, k\in\mathbb{N}$ with $s_0^s = 0$. Due to the communication bandwidth limitation, on the slave side, only one slave's data are allowed to be transmitted to the master by the TOD protocol at each instant $s_{k}^s, k\in\mathbb{N}$. Thus, for every $k\in\mathbb{N}$, the active slave obtaining the access to the communication network shall satisfy (\ref{eq:weighting_matrices}).

Following \cite{wei2021velocity}, we construct a virtual system for each manipulator to facilitate the adaptive control design with sampled-data communication.
Specifically, for the master manipulator, we construct the following virtual system: 
\begin{align}
		\ddot{x}_m=&-\alpha_m\dot{x}_m-\kappa_m\left(x_m-q_m\right)\nonumber\\
		&-\beta_m\left(x_m-\frac{1}{N}\sum_{i=1}^{N}\bar{\hat{x}}_{si}(t-T_s)\right)\label{virtual_m_c}
\end{align}
where $x_m\in\mathbb{R}^n$ is the master observer's output, $\bar{\hat{x}}_{si}$ is the reconstructed signal of the transmitted data from the $i$th slave satisfying 
\begin{align}\label{eq:reconstructed_hat_xsi}
	\bar{\hat{x}}_{si}=\hat{x}_{si}(s_{k}^s), t\in[s_{k}^s, s_{k+1}^s)
\end{align}
with $\hat{x}_{si}$ given in (\ref{eq:updating_x_sk}), 
and $\alpha_m, \kappa_m, \beta_m$ are positive constants. 
Similarly, the $i$th slave's virtual observer is constructed as follows:  
\begin{align}\label{virtual_si_c}
		\ddot{x}_{si}=&\!-\!\alpha_{si}\dot{x}_{si}\!-\!\kappa_{si}\!\left(x_{si}\!-\!q_{si}\right)
		\!-\!\beta_{si}\left(\check{x}_{si}\!-\!\bar{x}_m(t-T_m)\right)
\end{align}
where $x_{si}\in\mathbb{R}^n$ is the output of the $i$th slave observer, $\check{x}_{si}=x_{si}-\gamma_i$, $\bar{x}_{m}$ is the reconstructed signal of the master observer's output $x_m$ and satisfies:
\begin{align}\label{eq:bar_x_m}
	\bar{x}_m = x_m(s_{k}^s), t\in[s_{k}^s,s_{k+1}^s)
\end{align}
$\alpha_{si}, \kappa_{si}, \beta_{si}, i = 1, 2, ..., N$ are positive constants.


Denote by $e_z$ the error between the actual joint position and the virtual system's output $x_z$:
\begin{equation}\label{errors_v}
	e_z \triangleq q_z-x_z, z \in \mathcal{M}
\end{equation}
where $\mathcal{M}=\{m, s1, ...., sN\}$. The following synchronization variables are proposed:
\begin{equation}\label{s_variable}
	r_z \triangleq \dot{q}_z+\lambda_z e_z, z \in \mathcal{M}
\end{equation}
\black{where $\gamma_z>0$ is a constant, }
and in accordance with Property~P6, it is obtained 
\begin{align}\label{eq:Y_z1}
	&Y_z(q_z\!,\! \dot{q}_z\!,\! e_z\!,\! \dot{e}_z)\theta_z\!=\! \lambda_z M_z(q_z)\dot{e}_z\!+\!\lambda_z{C}_z(q_z, \dot{q}_z)e_z\!-\!{G}_z
\end{align}
with $z\in\mathcal{M}$, $Y_z(q_z, \dot{q}_z, e_z, \dot{e}_z)$ is the regressor related to the variables $q_z, \dot{q}_z, e_z, \dot{e}_z$.

To characterize the event-triggered control mechanism,  let the triggering time sequence of the controllers be 
$0 =t_0<t_1<...< t_{k_z}...<\infty$, where $t_{k_z}$ denotes the $k_z$th event triggering instant of the manipulator $z$'s controller.

The event-triggered adaptive controllers for the master and the slaves are designed as
\begin{eqnarray}\label{tol_c}
	\left\{
	\begin{aligned}
		\tau_z&=-\kappa_zr_z(t_{k_z})\!-\!Y_z(t_{k_z})\hat{\theta}_z(t_{k_z}),  \\
		\dot{\hat{\theta}}_z&=\Gamma_z Y^{T}_z r_z, t\in[t_{k_{z}}, t_{k_z+1}), z\in\mathcal{M}
			\end{aligned}
		\right.
\end{eqnarray}
where  $Y_z\triangleq 	Y_z(q_z, \dot{q}_z, e_z, \dot{e}_z)$ is speicified in (\ref{eq:Y_z1}), $\kappa_z>0$ and $0< \Gamma_z \in \mathbb{R}^{n\times n}$, $\hat{\theta}_z$ is the estimate of $\theta_z$.  It is clearly observed that due to the event-triggered mechanism, the control inputs are held as constant for $t \in [t_{k_z}, t_{k_z+1})$. 

Before proposing the event-triggered conditions, we first define the triggering functions:
\begin{align}\label{eq:event_iota1}
\iota_z\!=\!|\rho_z|\!+\!\kappa_z|\varepsilon_z|\!-\! \dfrac{\gamma_z}{2}|r_z|\!-\!\epsilon_z e^{\!-\!\nu_z t}\!,\! z\!\in\!\mathcal{M}
\end{align}
where  $\gamma_z>0, \epsilon_z>0,$ and $\nu_z>0$ are constants, and  for $t \in [t_{k_z}, t_{k_z+1})$,
\begin{align}
\rho_z&=Y_z\hat{\theta}_z-Y_z(t_{k_{z}})\hat{\theta}_z(t_{k_{z}})\label{eq:Y_error}\\
\varepsilon_z&=r_z-r_z(t_{k_{z}})\label{eq:s_error}
\end{align}
Thus, the triggering condition is given  as follows:
\begin{align}\label{eq:event1}
	t_{k_z+1}=\inf\{t|t>t_{k_z}, \iota_z\geq 0\}
\end{align}

According to (\ref{eq:updating_x_sk}) and (\ref{eq:reconstructed_hat_xsi}), one has 
\begin{eqnarray}\label{eq:updating_x_sk1}
\bar{\hat{x}}_{si}=\left\{\begin{array}{l}{\bar{x}_{si}, i=i_{k}^{*}} \\ {\bar{x}_{si}+\eta_i, i \neq i_{k}^{*}}\end{array}\right.
\end{eqnarray}
where 
\begin{align}\label{eq:bar_x_si}
	\bar{x}_{si} &\triangleq x_{si}(s_{k}^s),  t\in[s_{k}^s, s_{k+1}^s)
\end{align}
 is the reconstructed signal from the sampling signal $x_{si}(s_{k}^s)$. 
Hence, one has for $t\in[s_{k}^s, s_{k+1}^s)$, 
\begin{align}
\bar{\hat{x}}_{s i}\left( t-T_s\right)
	=&\left\{\begin{array}{l}{\bar{x}_{s i}\left(t-T_s\right), i=i_{k}^{*}} \\ {\bar{x}_{si}(t-T_s)+\eta_i(t-T_s), i \neq i_{k}^{*}}\end{array}\right.\nonumber\\
		=&\left\{\begin{array}{l}{x_{s i}\left(t-D_s\right), i=i_{k}^{*}} \\ {{x}_{si}(t-D_s)+\eta_i(t-T_s), i \neq i_{k}^{*}}\end{array}\right.
		\label{eq:bar_x_si2}
\end{align}
where $D_s \triangleq t-s_{k}^s+T_s, t\in[s_{k}^s, s_{k+1}^s)$. Similarly, based on $\bar{x}_m=x_m(s_{k}^s), t\in[s_{k}^s, s_{k+1}^s)$, one has 
\begin{align}
	\bar{x}_m(t-T_m)=x_m(t-D_m), t\in[s_{k}^s, s_{k+1}^s) \label{eq:bar_x_m2}
\end{align}
with $D_m \triangleq  t-s_{k}^s+T_m, t\in[s_{k}^s, s_{k+1}^s)$. It is easily inferred from Assumptions \ref{amp:sampling}-\ref{amp:delay} that 
\begin{align} \label{D}
	0\leq D_j\leq h+d_j\triangleq h_M^j, j=m, s
\end{align}

%

Hence, 
based on (\ref{eq:master})-(\ref{eq:slave}), (\ref{tol_c}),  (\ref{eq:event_iota1})-(\ref{eq:event1}),  (\ref{eq:bar_x_si2}), and (\ref{eq:bar_x_m2}), the closed-loop system can be expressed as  
\begin{equation}\label{closed-loop_v_c}
	\left\{
	\begin{aligned} 
		f_m=&M_m(q_m) \dot{r}_m\!+\!C_m(q_m\!,\!\dot{q}_m) r_m\!+\!\kappa_m r_m
		-Y_m \tilde{\theta}_m\\&-\rho_m-\kappa_m\varepsilon_m\\ 
		f_{si}=&M_{si}(q_{si}) \dot{r}_{si}\!+\!C_{si}(q_{si}\!,\!\dot{q}_{si}) r_{si}\!+\!\kappa_{si} r_{si}-Y_{si} \tilde{\theta}_{si}\\&-\rho_{si}-\kappa_{si}\varepsilon_{si}\\
		\ddot{x}_m=&-\!\alpha_m\dot{x}_m\!-\!\beta_m\left(x_m\!-\!\frac{1}{N}\sum_{i=1}^{N}\!{x}_{si}(t\!-\!D_s)\right)\\
		&+\kappa_m e_m+\frac{\beta_m}{N}\sum_{i=1,i\neq i_k^*}^N\eta_{i}(t-T_s)\\
		\ddot{x}_{si}=&-\alpha_{si}\dot{x}_{si}-\beta_{si}\bigg(\check{x}_{si}-{x}_m(t-D_m)\bigg)\!+\!\kappa_{si}e_{si}\\
		\dot{\eta}_i=&0, t\in[s_{k}^s, s_{k+1}^s), i = 1,2,...,N
	\end{aligned}\right.
\end{equation}

where
 	$\tilde{\theta}_z=\theta_z-\hat{\theta}_z$, $z\in\mathcal{M}$.

By using (\ref{eq:bar_x_si2}), one obtains for $i=i_k^*\in\{1, 2, ..., N\}$,
\begin{align}
	\eta_{i}(s_{k+1}^s) \!=\! \hat{x}_{si}(s_{k}^s)-x_{si}(s_{k+1}^s)\!=\! x_{si}(s_{k}^s)-x_{si}(s_{k+1}^s)
\end{align}
	and for $i\neq i_k^*$, $i\in \{1, 2, ..., N\}$, 
	\begin{align}
\eta_{i} (s_{\!k+1\!}^s)\!=\!\hat{x}_{\!si\!}(s_{\!k\!}^s)\!-\!x_{\!si\!}(s_{\!k\!+\!1\!}^s)\!=\!\eta_i({s_{k}^s})\!+\!x_{\!si\!}(s_{k}^s)\!-\!x_{\!si\!}(s_{\!k\!+\!1\!}^s)
	\end{align}

Thus,  following \cite{li2019bilateral}, the delayed reset system is given by 
\begin{align}\label{eq:event_c_reset}
	\left\{
	\begin{aligned}
&q_{si}(s_{k+1}^s) = q_{si}(s_{k+1}^{s-})\\
&x_{si}(s_{k+1}^s) = x_{si}(s_{k+1}^{s-})\\
&\eta_{i} = [1-\delta(i, i_k^*))]\eta_i(s_{k}^s)+x_{si}(s_{k}^s)-x_{si}(s_{k+1}^s)\\
&k\in\mathbb{N}, i = 1, 2, ..., N
\end{aligned}
\right.
\end{align}
where $\delta$ is the Kronecker delta. 

Following the typical practice, we first define the Lyapunov-like function $\mathcal{V}_z=1/2r_z^TM_z(q_z, \dot{q}_z)r_z+1/2\tilde{\theta}_z^T \Gamma_z^{-1}\tilde{\theta}_z, \forall z \in \mathcal{M}$, where $0<\Gamma_z\in\mathbb{R}^{n\times n}$. The derivative of $V_z$ along the trajectories of the system (\ref{closed-loop_v_c}) can be written as (using Property~P2) $\dot{\mathcal{V}}_z=-\kappa_zr_z^Tr_z+r_z^T(\rho_z+\kappa_z\varepsilon_z+f_z) $.
Then,  the following Lyapunov-Krasovskii functional for the closed-loop system (\ref{closed-loop_v_c}) and (\ref{eq:event_c_reset}) is proposed  with $t \in [s_{k}^s, s_{k+1}^s)$:
\begin{equation}\label{V_v_c}
V=V_1+V_2+V_T
\end{equation}
where
\begin{align}
V_1=&\frac{N}{2\lambda_m}\mathcal{V}_m\!+\!\frac{N}{2}\kappa_me_m^Te_m\!+\!\frac{N}{2}\dot{x}_m^T\dot{x}_m+\frac{1}{2}\sum_{i=1}^{N}\bigg(\!\frac{\beta_m}{\beta_{si}\lambda_{si}}\!\mathcal{V}_{si}\nonumber\\
&\!+\!\frac{\beta_m}{\beta_{si}}\dot{x}_{si}^T\dot{x}_{si}\!+\!\beta_m(x_m\!-\!\check{x}_{si})^T\!(x_m\!-\!\check{x}_{si})\!+\!\frac{\beta_m\!\kappa_{si}}{\beta_{si}}e_{si}^T\!e_{si}\!\bigg)\nonumber
\end{align}
\begin{align*}
V_2=&N\int_{-h_M^m}^{0}\int_{t+v}^{t}\dot{x}_m^T(\sigma)R_m\dot{x}_m(\sigma)d\sigma dv\\
&+\sum_{i=1}^{N} \int_{-h_M^s}^{0}\int_{t+v}^{t}\dot{x}_{si}^T(\sigma)R_{si}\dot{x}_{si}(\sigma)\,d\sigma dv\\
		V_T=&\sum_{i\!=\!1, i \!\neq\! i_k^{*}}^{N}\left(\frac{1}{h^2}\int_{t-T_s}^t\eta_i^T(v)Z_i\eta_i(v)dv\!+\!\frac{s_{k}^s-t}{s_{k+1}^s\!-\!s_{k}^s} \eta_{i}^{T} U_{i} \eta_{i}\right)\\
		&+\sum_{i=1}^{N} \left(h \int_{s_{k}^s}^{t} \dot{x}_{si}^T(v) P_i \dot{x}_{si}(v)\,dv+\eta_i^TQ_i\eta_{i}\right)\nonumber
\end{align*}
with $R_m, R_{si}, Z_i, P_i, U_i, Q_i\in\mathbb{R}^{n\times n}$ as positive definite matrices. 

It is observed that for $t \in\left[s_{k}^s, s_{k+1}^s\right)$, the Lyapunov function $V$ is continuous and differentiable. At  \textcolor{black}{the instants $s_{k+1}^s, k\in\mathbb{N}$}, it can also be justified that $V$ is positive and doesn't grow, which is guaranteed by the following Lemma. 

\begin{Lemma}\label{LMIs_v}
	If there exist positive-definite matrices $Q_{i} \in \mathbb{R}^{n \times n}$, $U_{i} \in \mathbb{R}^{n \times n}$ and $P_{i} \in \mathbb{R}^{n \times n}$ such that the following linear-matrix-inequalities (LMIs) hold:
	\begin{align} \label{LMI_v}
	&\Omega_{i}\!=\!\left[\begin{array}{cc}{-\frac{1}{N-1} Q_{i}\!+\!U_{i}} & {Q_{i}} \\ {*} & {-P_{i}+Q_{i}}\end{array}\right]\!<\!0,  i\!=\!1, \ldots, N
	\end{align}
	then $V$ is positive in the sense that 
	\begin{align} \label{positivity of Ve_v}
	V\geq& a (|\dot{q}_m|^2\!+\!|\dot{q}_s|^2\!+\!|\dot{x}_m|^{2}\!+\!|\dot{x}_s|^{2}\!+\!|e_m|^{2}\!+\!|e_s|^{2}\!
	\nonumber\\&
	+\!|\tilde{\theta}_m|^2\!+\!|\tilde{\theta}_s|^2\!+\!|e|^{2})
	\end{align}
	where $\dot{q}_s \triangleq \operatorname{col}\left\{\dot{q}_{s1}, \ldots, \dot{q}_{sN}\right\}$, $\dot{x}_s \triangleq \operatorname{col}\left\{\dot{x}_{s1}, \ldots, \dot{x}_{sN}\right\}$, $e_s \triangleq \operatorname{col}\left\{e_{s1}, \ldots, e_{sN}\right\}$, $e \triangleq \operatorname{col}\left\{e_{1}, \ldots, e_{N}\right\}$ with $e_{i}=x_{m}-\check{x}_{si}$ for some $a>0$. Moreover, $V_e$ doesn't grow at  $s_{k+1}^s, k \in\mathbb{N}$:
	\begin{equation}\label{eq:Theta_v}
	V\left(s_{k+1}^s\right)-V\left(s_{k+1}^{s-}\right) \leq 0
	\end{equation}
	\begin{proof}
		The proof follows the same line of reasoning as that of Lemma 1 in \cite{li2019bilateral}.

	\end{proof}
\end{Lemma}

The main result on the stability of the closed-loop system (\ref{closed-loop_v_c}) is summarized in the following theorem.

\begin{Theorem} \label{thm:Theorem_c}
	Consider the closed-loop teleoperation system (\ref{closed-loop_v_c})-(\ref{eq:event_c_reset}) with the event-triggered adaptive controller (\ref{tol_c}) under the TOD scheduling protocol (\ref{eq:weighting_matrices}). If there exist positive-definite matrices $R_{m}, R_{si}, P_{i}, U_{i}, Q_i, Z_i$,  $i = 1, 2, \ldots, N$ such that the LMIs (\ref{LMI_v}) and 
	\begin{align} \label{lmi_c}
			&\Xi_i
			\!=\!\begin{bmatrix}
			 {\Xi_{i}^{11}} & {0}  & {0} &\Xi_{i}^{14}& {0} & {\Xi_i^{16}} \\
			 {*} & {\Xi_{i}^{22}}  & \Xi_i^{23} & {0} & {0} & {0} \\ 
		 {*} & {*} &{\Xi_{i}^{33}} & {0} & {0} & {0} \\ 
		 {*} & {*} & {*} & {\Xi_{i}^{44}} & {0} & {0} \\ 
		 {*} & {*} & {*} & {*} & {\Xi_i^{55}} & {0} \\  
		{*} & {*}  & {*} & {*} & {*} & {\Xi_i^{66}}
			\end{bmatrix}
			<0
			\\&
			\black{\kappa_{m}>\gamma_m, \kappa_{si}>\gamma_{si}}, i = 1, 2, \dots, N\label{eq:lmi_control_2}
		\end{align}
	are satisfied,  	where
			\begin{align*}
				\Xi_{i}^{11}&=-\alpha_mI+h_M^mR_m, \Xi_i^{14}=-\frac{\beta_mh_M^s}{2}I\\
	\Xi_i^{16}&=\operatorname{row}_{j=1, \ldots, N}\left\{\frac{\beta_mh}{2}I, j \neq i\right\}\\
				\Xi_{i}^{22}&=-\frac{\beta_m\alpha_{si}}{\beta_{si}}I+h_M^sR_{si}+hP_i, \Xi_i^{23}=-\frac{\beta_mh_M^m}{2}I\\
				\Xi_{i}^{33}&=-h_M^mR_m, \Xi_{i}^{44}=-h_M^sR_{si}\\ 
				\Xi_i^{55}&=-\operatorname{diag}_{j=1, \ldots, N}\left\{hU_{j}-Z_j, j \neq i\right\}\\
				\Xi_i^{66}&=- \operatorname{diag}_{j=1, \ldots, N}\left\{(1-p_s)Z_j, j \neq i\right\}\nonumber
			\end{align*}

 then the following claims hold:
		\begin{enumerate}
			\item[(1)] if the considered teleoperation system is in free motion, i.e. the human force $f_m$ and the environment forces $f_{si}$ are zero, then all the signals are bounded, the system (\ref{closed-loop_v_c}) is asymptotically stable and $\lim_{t \rightarrow \infty}\left(x_m-\check{x}_{si}\right)=\lim_{t \rightarrow \infty}\dot{x}_m=\lim_{t \rightarrow \infty}\dot{x}_{si}=\lim_{t \rightarrow \infty}\left(q_m-\check{q}_{si}\right)=\lim_{t \rightarrow \infty}\dot{q}_m=\lim_{t \rightarrow \infty}\dot{q}_{si}=0$, $i=1,\ldots, N$, $\lim_{t \rightarrow \infty}e_m =\lim_{t \rightarrow \infty}e_{si}=0 $;
			\item[(2)] if $f_m, f_{si} \in \mathcal{L}_{2} \cap \mathcal{L}_{\infty} $, then all the signals are bounded, and $\lim_{t \rightarrow \infty}\left(x_m-\check{x}_{si}\right)=\lim_{t \rightarrow \infty}\dot{x}_m=\lim_{t \rightarrow \infty}\dot{x}_{si}=\lim_{t \rightarrow \infty}\left(q_m-\check{q}_{si}\right)=\lim_{t \rightarrow \infty}\dot{q}_m=\lim_{t \rightarrow \infty}\dot{q}_{si}=\lim_{t \rightarrow \infty}e_m =\lim_{t \rightarrow \infty}e_{si}=0$;
			\item[(3)] Zeno behavior is excluded in the presented control scheme.
	\end{enumerate}
\end{Theorem}

\begin{proof}
	
	Firstly, suppose that the manipulators are in free motion. The derivative of $V_1, V_2$ along the trajectory of the system for $t>0$ are 
	\begin{align}
	\dot{V}_1=&-\frac{N\kappa_m}{2\lambda_m}\dot{q}_m^T\dot{q}_m-\frac{N\kappa_m\lambda_m}{2}e_m^Te_m-N\alpha_m\dot{x}_m^T\dot{x}_m \nonumber\\
	&+\sum_{i=1}^{N}\!\bigg(\!\!-\!\frac{\beta_m\!\kappa_{si}}{2\beta_{si}\!\lambda_{si}}\dot{q}_{si}^T\dot{q}_{si}\!-\!\frac{\beta_m\!\kappa_{si}\!\lambda_{si}}{2\beta_{si}}e_{si}^Te_{si} 
	\!-\!\frac{\beta_m\!\alpha_{si}}{\beta_{si}}\dot{x}_{si}^T\!\dot{x}_{si}\!\bigg)
	\nonumber\\&
	-\beta_m\sum_{i=1}^{N}\!\bigg(\!\dot{x}_{si}^T L_m+\dot{x}_m^T L_{si}\bigg)\!+\!\frac{N}{2\lambda_m}r_m^T(\rho_m\!+\!\kappa_m\varepsilon_m) \nonumber\\
	&\!+\!\sum_{i=1}^{N}\!\!\frac{\beta_m}{2\beta_{si}\!\lambda_{si}}r_{si}^T(\rho_{si}\!+\!\kappa_{si}\varepsilon_{si})
	\!+\!\beta_m\dot{x}_m^T\!\!\sum_{i\!=\!1,i\!\neq\! i_k^*}^N\!\!\eta_{i}(t\!-\!T_s)\nonumber\\
	\dot{V}_2 \le &Nh_M^m\dot{x}_m^T R_m \dot{x}_m-\frac{N}{h_M^m}L_m^TR_mL_m\nonumber \\
	&+\sum_{i=1}^{N}\left(h_M^s\dot{x}_{si}^T R_{si}\dot{x}_{si}-\frac{1}{h_M^s}L_{si}^TR_{si}L_{si}\right) \nonumber 
	\end{align}
where 
$ L_{m} \!=\!\int_{t\!-\!D_m}^{t}\dot{x}_m(v)dv$,
$L_{si}\! =\!\int_{t\!-\!D_{s}}^{t}\dot{x}_{si}(v)dv$.

	For $t\in[s_{k}^s, s_{k+1}^s)$, one has 
	\begin{align}
		\dot{V}_T\leq& \sum_{i=1, i \neq i_{k}^{*}}^N\eta_i^T\left(-\frac{U_i}{h}+\frac{Z_i}{h^2}\right)\eta_i+h	\sum_{i=1}^N\dot{x}_{si}P_i\dot{x}_{si}\nonumber\\
		&\!+\!\sum_{i=1, i \neq i_{k}^{*}}^N\frac{1-p_s}{h^2}\eta_i^T(t-T_s\!)Z_i\eta_i(t-T_s)
	\end{align}
	According to the proposed event-triggered condition (\ref{eq:event1}), we get that $|\rho_{z}|+\kappa_{z}|\varepsilon_{z}|\black{\leq} {\gamma_{z}}|r_{z}|/2+\epsilon_{z} e^{-\nu_{z} t}$ for \black{$t\geq 0, z\in\mathcal{M}$}, Thus, applying the Young's inequality, one has 
	\black{
	 \begin{align}
	 r_z^T(\rho_z\!+\!\kappa_z\varepsilon_z)
	 \leq &|r_z|(|\rho_z|+\kappa_z|\varepsilon_z|)\nonumber\\
	 \leq &(\gamma_z+\bar{\gamma}_z)\left(|\dot{q}_z|^2\!+\!\lambda_z^2|e_z|^2\right)\!+\!\frac{\epsilon_z^2}{2\bar{\gamma}_z}e^{\!-\!2\nu_z t}
	 \end{align}}
\black{for any $\bar{\gamma}_z>0$, $t\geq 0$, and $z\in\mathcal{M}$}. 

\black{Denote $
				\phi_{m}\triangleq\frac{\kappa_m-\gamma_m-\bar{\gamma}_m}{2\lambda_m}, \phi_{si}\triangleq \frac{\beta_m(\kappa_{si}-\gamma_{si}-\bar{\gamma}_{si})}{2\beta_{si}\lambda_{si}},							\psi_{m} \triangleq\lambda_m^2\phi_m, \psi_{si}\triangleq\lambda_{si}^2\phi_{si}$, 
Therefore, we obtain that  $\phi_{m}, \phi_{si}, \psi_{m}, \psi_{si}>0$ from (\ref{eq:lmi_control_2}).} Hence, for $t\in[s_{k}^s, s_{k+1}^s)$,  one has
\begin{align} \label{eq:V_e_final}
	\dot{V} \!\le\!& \sum_{i=1}^{N}\!\bigg( \!-\!\phi_{m}\dot{q}_m^T\dot{q}_m\!-\!\phi_{si}\dot{q}_{si}^T\dot{q}_{si}\!-\!\psi_{m}e_m^Te_m
	-\psi_{si}e_{si}^Te_{si}\nonumber\\
	&+\chi_{i}^T\underline{\Xi}\Xi_i\underline{\Xi}\chi_i+\frac{\epsilon_m^2 e^{-2\nu_m t}}{4\lambda_m\gamma_m}+\frac{\beta_m\epsilon_{si}^2 e^{-2\nu_{si}t}}{4\beta_{si}\lambda_{si}\gamma_{si}}\bigg)
\end{align}
where  $\underline{\Xi}=\text{diag}\{1,1,1/h_M^m, 1/h_M^s,1/h,1/h\}\black{\otimes I}$,

$
\chi_{i}=\operatorname{col}\left\{ \dot{x}_m, \dot{x}_{si},  L_m, L_{si}, \varrho_{i_{k}^*}, \varsigma_{i_{k}^*}\right\}
$
, 
$\varrho_{i}=\operatorname{col}_{j=1,2, \ldots, N}\left\{\eta_{j}, j \neq i\right\}$, $
\varsigma_{i}=\operatorname{col}_{j=1,2, \ldots, N}\left\{\eta_{j}(t-T_s), j \neq i\right\}$.

Let $\xi_i\triangleq-\lambda_{\max}(\underline{\Xi}\Xi_i\underline{\Xi})$. According to (\ref{lmi_c}), one has $\xi_i > 0$. Therefore, $\dot{V}$ can be rewritten for $t\in[s_k^s,s_{k+1}^s)$ as
\begin{align}\label{V_e}
\dot{V}\! \le\!& \sum_{i=1}^{N}\bigg( \!-\!\phi_{m}|\dot{q}_m|^2-\phi_{si}|\dot{q}_{si}|^2\!-\!\xi_i|\chi_{i}|^2\!-\!\psi_{m}|e_m|^2\!\!-\!\!\psi_{si}|e_{si}|^2\nonumber\\
&+\frac{\epsilon_m^2 e^{-2\nu_m t}}{4\lambda_m\black{\bar{\gamma}_m}}+\frac{\beta_m\epsilon_{si}^2 e^{-2\nu_{si}t}}{4\beta_{si}\lambda_{si}\black{\bar{\gamma}_{si}}}\bigg)
\end{align}
Integrating both sides of (\ref{V_e}) from $s_k$ to $t$, $t \in [s_k^s, s_{k+1}^s)$,  one has
\begin{align}\label{vet1}
V \leq & V(s_k^s)-\sum_{i=1}^{N}  \int_{s_k^s}^{t}\left(\phi_{m}|\dot{q}_m(\sigma)|^2+\phi_{si}|\dot{q}_{si}(\sigma)|^2\right)\, d\sigma\nonumber\\
&-\sum_{i=1}^{N}  \int_{s_k^s}^{t}\left(\psi_{m}|e_m(\sigma)|^2+\psi_{si}|e_{si}(\sigma)|^2+\xi_i|\chi_{i}(\sigma)|^{2}\right)\, d\sigma\nonumber\\&+\sum_{i=1}^N \int_{s_k}^{t}\left(\frac{\epsilon_m^2 e^{-2\nu_m \sigma}}{4\lambda_m\black{\bar{\gamma}_m}}+\frac{\beta_m\epsilon_{si}^2 e^{-2\nu_{si} \sigma}}{4\beta_{si}\lambda_{si}\black{\bar{\gamma}_{si}}}\right)\,d\sigma
\end{align}
Substituting $t=s_{k+1}^{s^{-}}$ into (\ref{vet1}), and using (\ref{eq:Theta_v}), we get
\begin{align} \label{vet2_v1}
V(s_{k+1}^s) \leq & V(0)\!-\!\sum_{i=1}^{N}  \int_{0}^{s_{k+1}^s}\left(\phi_{m}|\dot{q}_m(\sigma)|^2\!+\!\phi_{si}|\dot{q}_{si}(\sigma)|^2\right)\, d\sigma\nonumber\\
&-\sum_{i=1}^{N}  \int_{0}^{s_{k+1}^s}\left(\psi_{m}|e_m(\sigma)|^2+\psi_{si}|e_{si}(\sigma)|^2\right)\, d\sigma\nonumber\\
&+ \sum_{i=1}^{N} \int_{0}^{s_{k+1}^s}\left(\frac{\epsilon_m^2 e^{-2\nu_m \sigma}}{4\lambda_m\black{\bar{\gamma}}_m}+\frac{\epsilon_{si}^2 e^{-2\nu_{si} \sigma}}{4\lambda_{si}\black{\bar{\gamma}}_{si}}\right)\,d\sigma\nonumber\\
&-\sum_{i=1}^{N}  \int_{0}^{s_{k+1}^s}\xi_i|\chi_{i}(\sigma)|^{2}\, d\sigma
\end{align} 

Finally, based on (\ref{vet1}) and (\ref{vet2_v1}), we have for all $t\geq 0$

\begin{align} \label{vet4_v}
V \!\leq\! & V(0)\!+\! \Delta_1\!-\!\sum_{i=1}^{N} \!\! \int_{0}^{t}\!\!\left(\!\phi_{m}|\dot{q}_m(\sigma)|^2\!+\!\phi_{si}|\dot{q}_{si}(\sigma)|^2\!\right)\!\, d\sigma\nonumber\\
&\!-\!\sum_{i=1}^{N} \!\! \int_{0}^{t}\!\!\left(\!\psi_{m}|\!e_m(\sigma)\!|^2\!+\!\psi_{si}|\!e_{si}(\sigma)\!|^2\!+\!\xi_i|\!\chi_{i}(\sigma)\!|^{2}\!\right)\!\, d\sigma
\end{align}
with
$
\Delta_1=\sum_{i=1}^{N} \int_{0}^{t}\left(\frac{\epsilon_m^2e^{-2\nu_m \sigma}}{4\lambda_m\black{\bar{\gamma}}_m}+\frac{\beta_m\epsilon_{si}^2 e^{-2\nu_{si} \sigma}}{4\beta_{si}\black{\bar{\gamma}}_{si}\gamma_{si}}\right)\,d\sigma \nonumber
$.

Thus, (\ref{vet4_v}) clearly implies that $V \in \mathcal{L}_{\infty}$ and $\dot{q}_m$, $\dot{q}_{si}$, $e_m$, $e_{si}$, $\chi_{i} \in \mathcal{L}_{2}$ which means $ \dot{x}_m, \dot{x}_{si}, \eta_{i} \in \mathcal{L}_{2}$. $V \in \mathcal{L}_{\infty}$ and the fact (\ref{positivity of Ve_v}) show that $\dot{q}_m, \dot{q}_{si}, \dot{x}_m, \dot{x}_{si}, e_m, e_{si}, x_m-\check{x}_{si} \in \mathcal{L}_{\infty}$, which further implies that $\eta_i\in\mathcal{L}_\infty$.  Since $\dot{e}_m=\dot{q}_m-\dot{x}_m \in \mathcal{L}_{\infty}$ and $\dot{e}_{si}=\dot{q}_{si}-\dot{x}_{si} \in \mathcal{L}_{\infty}$, then invoking standard Barbalat's Lemma, it is concluded that $e_m \rightarrow 0$ as $t \rightarrow \infty$ and $e_{si} \rightarrow 0$ as $t \rightarrow \infty$. 

Next, we proceed to verify that the master-slaves position synchronization of virtual systems (\ref{virtual_m_c})-(\ref{eq:bar_x_m}) can be achieved. By the fact that $\dot{q}_m, \dot{q}_{si}, \dot{x}_m, \dot{x}_{si}, e_m, e_{si}, x_m-\check{x}_{si}, \eta_i$ are bounded, we have $\ddot{q}_m, \ddot{q}_{si}, \ddot{x}_m, \ddot{x}_{si}$ are bounded for $ t \in [s_k^s, s_{k+1}^s)$ from the closed-loop teleoperation system (\ref{closed-loop_v_c}). Consequently, we conclude that $\dot{q}_m \rightarrow 0$ as $t \rightarrow \infty$, $\dot{q}_{si} \rightarrow 0$ as $t \rightarrow \infty$, $\dot{x}_m \rightarrow 0$ as $t \rightarrow \infty$ and $\dot{x}_{si} \rightarrow 0$ as $t \rightarrow \infty$ by the generalized Barbalat's Lemma \cite{li2021distributed,wang2020differential}. By the differentiation of $\ddot{x}_m$ and $\ddot{x}_{si}$ on both sides respectively, we have $\dddot{x}_m, \dddot{x}_{si}$ are bounded for $ t \in [s_k^s, s_{k+1}^s)$ with the boundedness of $\ddot{x}_m, \ddot{x}_{si}, e_m, e_{si}, x_m-\check{x}_{si}$. Therefore, based on the generalized Barbalat's Lemma, we know that $\ddot{x}_m \rightarrow 0$ as $t \rightarrow \infty$ and $\ddot{x}_{si} \rightarrow 0$ as $t \rightarrow \infty$. Then, from (\ref{closed-loop_v_c}), we have $x_m-\check{x}_{si} \rightarrow 0$ as $t \rightarrow \infty$, that is, $x_m-(x_{si}-\gamma_i) \rightarrow 0$ as $t \rightarrow \infty$. By the fact that $e_m \rightarrow 0$ as $t \rightarrow \infty$ and $e_{si} \rightarrow 0$ as $t \rightarrow \infty$, it is concluded that $q_m-\check{q}_{si} \rightarrow 0$ as $t \rightarrow \infty$.

	If $f_m, f_{si} \in \mathcal{L}_{2} \cap \mathcal{L}_{\infty}$, $\dot{V}$ can be rewritten as 
	\begin{align} \label{eq:V_e_final-f}
	\dot{V} \le& \sum_{i=1}^{N}\bigg( -\phi_{m}\dot{q}_m^T\dot{q}_m-\phi_{si}\dot{q}_{si}\dot{q}_{si}
	-\psi_{m}e_m^Te_m-\psi_{si}e_{si}^Te_{si}\nonumber\\&
	+\chi_{i}^T\Xi_i\chi_i+\frac{1}{2\lambda_m}r_m^Tf_m+\frac{\beta_m}{2\beta_{si}\lambda_{si}}r_{si}^Tf_{si}\nonumber\\
	&+\frac{\epsilon_m^2e^{-2\nu_m t}}{4\lambda_m\black{\bar{\gamma}}_m}+\frac{\beta_m\epsilon_{si}^2 e^{-2\nu_{si} t}}{4\beta_{si}\lambda_{si}\black{\bar{\gamma}}_{si}}\bigg)
	\end{align}

	Applying the Young's inequality to (\ref{eq:V_e_final-f}), we have
	\begin{align} \label{eq:fve2-f_c}
	\dot{V} \le& \sum_{i=1}^{N}\bigg( \!-\!\frac{\phi_{m}}{2}|\dot{q}_m|^2\!-\!\frac{\phi_{si}}{2}|\dot{q}_{si}|^2\!-\!\xi_i|\chi_{i}|^2\!-\!\frac{\psi_{m}}{2}|e_m|^2\nonumber\\
	&+\frac{1}{4\lambda_m^2\phi_m}|f_m|^2\nonumber+\frac{\beta_m^2}{4\lambda_{si}^2\phi_{si}\beta_{si}^2}|f_{si}|^2-\frac{\psi_{si}}{2}|e_{si}|^2\nonumber\\
	&+\frac{\epsilon_m^2e^{-2\nu_m t}}{4\lambda_m\black{\bar{\gamma}}_m}+\frac{\beta_m\epsilon_{si}^2 e^{-2\nu_{si} t}}{4\beta_{si}\lambda_{si}\black{\bar{\gamma}}_{si}}\bigg)
	\end{align} 


	Similar to (\ref{vet4_v}), we have  

\begin{align} \label{vet3_v}
	V \leq & V(0)-\sum_{i=1}^{N}  \int_{0}^{t}\left(\frac{\phi_{m}}{2}|\dot{q}_m(\sigma)|^2+\frac{\phi_{si}}{2}|\dot{q}_{si}(\sigma)|^2\right)\, d\sigma\nonumber\\
	&-\sum_{i=1}^{N}  \int_{0}^{t}\left(\frac{\psi_{m}}{2}|e_m(\sigma)|^2+\frac{\psi_{si}}{2}|e_{si}(\sigma)|^2\right)\, d\sigma\nonumber\\
	&-\sum_{i=1}^{N}  \int_{0}^{t}\xi_i|\chi_{i}(\sigma)|^{2}\, d\sigma+ \Delta_1 + \Delta_2
	\end{align} 
	with 
$\Delta_2 = \sum_{i=1}^{N} \frac{1}{4\lambda_m^2\phi_m}\|f_m\|_2^2\nonumber+\sum_{i=1}^{N} \frac{\beta_m^2}{4\lambda_{si}^2\phi_{si}\beta_{si}^2} \|f_{si}\|_2^2 < \infty \nonumber
$.

The remainder of the proof is obvious by following the same line of reasoning for Claim (1). 

Finally, we show that Claim (3) is true. Let 
	\begin{align}
		\sigma_z \triangleq \rho_z+ \kappa_z\varepsilon_z, z\in\mathcal{M}
		\end{align}
	For $t \in [t_{k_z}, t_{{k_z}+1})$, calculating the upper right-hand Dini derivative of $|\sigma_z|$, we obtain that
\begin{align}\label{eq:Dini}
	D^+|\sigma_z| =\frac{\sigma_z^T\dot{\sigma}_z}{|\sigma_z|}\le |\dot{\sigma}_z|=|\kappa_z\dot{r}_z+Y_z\dot{\hat{\theta}}_z+\dot{Y}_z\hat{\theta}_z|
\end{align}
Using  Assumption~\ref{amp:Y} and the boundedness of $\ddot{q}_z, \dot{x}_z, \hat{\theta}_z, \dot{\hat{\theta}}_z$, it is concluded that there exists a constant $Q_z \in \mathbb{R}^+$ such that $|\kappa_z\dot{r}_z+Y_z\dot{\hat{\theta}}_z+\dot{Y}_z\hat{\theta}_z|\leq Q_z $. 
Thus, from (\ref{eq:event_iota1}), (\ref{eq:event1}) and (\ref{eq:Dini}),
the inter-event intervals satisfy: 
\begin{align}\label{eq:inter-event0}
	t_{{k_z}+1}-t_{k_z} \ge \dfrac{|\sigma_z(t_{{k_z}+1}^z)|-|\sigma_z(t_{k_z}^z)|}{Q_z}
\end{align}
From the triggering conditions (\ref{eq:event_iota1}) and (\ref{eq:event1}), one has $\lim\limits_{t\to t_{{k_z}+1}}|\sigma_z| \ge \epsilon_ze^{-\nu_z t}$. Combining with $\sigma_z(t_{k_z})=0$, (\ref{eq:inter-event0}) can be written as
\begin{align}\label{inter-event1}
	t_{{k_z}+1}-t_{k_z}&\ge \dfrac{\epsilon_ze^{-\nu_z t}}{Q_z}\ge \dfrac{\epsilon_ze^{-\nu_z t_{k_z+1}}}{Q_z} >0
\end{align}

Assume, for contradiction, that Zeno behavior occurs and $\lim_{k \to \infty} t_{k_z} = \mathcal{T}_z$ for some finite $\mathcal{T}_z > 0$, $z\in\mathcal{M}$. Then, for any $t \in [0, \mathcal{T}_z)$, there exists a constant $g_z > 0$ such that
\[
g_z \leq \frac{\epsilon_z e^{-\nu_z t}}{Q_z}.
\]
Hence, the cumulative time up to the $k_z$th event satisfies:
\begin{align}
	t_{k_z} &= (t_{k_z} - t_{k_z-1}) + (t_{k_z-1} - t_{k_z-2}) + \cdots + (t_{1_z} -0) \notag \\
	&\geq k_z \cdot \frac{\epsilon_z e^{-\nu_z t}}{Q_z} \geq k_z g_z.
\end{align}
Taking the limit as $k_z \to \infty$, we obtain:
\[
\lim_{k_z \to \infty} t_{k_z} \geq \lim_{k_z \to \infty} k_z g_z = \infty,
\]
which contradicts the assumption that $\lim_{k_z \to \infty} t_{k_z} = \mathcal{T}_z$. Thus, Zeno behavior is excluded by contradiction.
This completes the proof.

\end{proof}

\begin{remark}
	In this section, a novel hybrid time/event-triggered interaction framework for teleoperation is proposed. The time-triggered sampling mechanisms are designed for the inter-robot information interaction, while the control updating is driven by an event-triggered mechanism.
\end{remark}

\begin{remark}\label{rek:mininum_IET}
		For the event-triggered control, even though it is proven that Zeno behavior does not exist due to the triggering functions in  (\ref{eq:event1}), it can still be observed that the minimum inter-event interval of each manipulator tends to \black{be quite small} as  $t\to\infty$ from (\ref{inter-event1}). This drawback should be removed in the future work.
\end{remark}


We summarize the event-triggered control in Algorithm~\ref{alg:control}.
\begin{algorithm}[H]
\caption{Adaptive Event-triggered control with the TOD scheduling protocol}\label{alg:control}
	\hspace*{0.00in}{\bf Input:} Measurements $q_z, \dot{q}_z, x_z, \dot{x}_z$ from the controlled manipulator $z$, $z\in\mathcal{M}$; 
 virtual systems (\ref{virtual_m_c})  and (\ref{virtual_si_c});\\
	\hspace*{0.00in}{\bf Output: }Event time $t_{k_z}$, $k_z= 1, 2, \cdots$, and $\tau_z(t_k)$ for the manipulator $z$, $z\in\mathcal{M}$;
	\begin{algorithmic}
		\FORALL{$z\in\mathcal{M}$}
		\WHILE{$t< t_{k_z}$}
		\STATE measure $q_z, \dot{q}_z$;
		\STATE receive $x_j$ from the remote side ($j=m$ if $z = {si}$ and $j={si}$ if $z=m$, $i= 1, \cdots, N$);
		\STATE compute  $ x_z, \dot{x}_z$ by evolving the virtual systems (\ref{virtual_m_c})  if $z = m$ and (\ref{virtual_si_c})  if $z=si$;
		\IF {$t=s_{k_z}^z$}
		\IF{$z=m$}
		\STATE transmit $x_m(s_{k_m}^m)$ to the slave side;
		\ELSE 
		\STATE pick $i_k^*$ by (\ref{eq:weighting_matrices}) and transmit $x_{si}(s_{k_{si}}^{si})$ with $i = i_k^*$ to the master side;
		\ENDIF
		\ENDIF
		\STATE compute $e_z, \dot{e}_z$ from (\ref{errors_v}); compute $r_z$ from (\ref{s_variable}); compute $Y_z$ from (\ref{eq:Y_z1});
		\STATE compute $\hat{\theta}_z$ by integrating $\dot{\hat{\theta}}_z$ from $\hat{\theta}_z(0)$ using (\ref{tol_c});
		\STATE compute $\iota_z$ from (\ref{eq:event_iota1})- (\ref{eq:s_error});
		\IF{$t_{k_z+1}$ is found with (\ref{eq:event1})}
		 \STATE  $k_z \leftarrow k_z+1 $; compute $\tau_{z}(t_{k_z})$ from (\ref{tol_c}); 
		   \ENDIF
		\ENDWHILE		
		\ENDFOR
	\end{algorithmic}
\end{algorithm}


\subsection{Event-triggered communication}
In this section, we address the control design of the SMMS teleoperation system (\ref{eq:master})-(\ref{eq:slave}) under an event-triggered communication. The control architecture is described in Fig.~\ref{fig:control setup_v}.

\begin{figure}[!htb]
	\centering
	\includegraphics[width=0.95\linewidth]{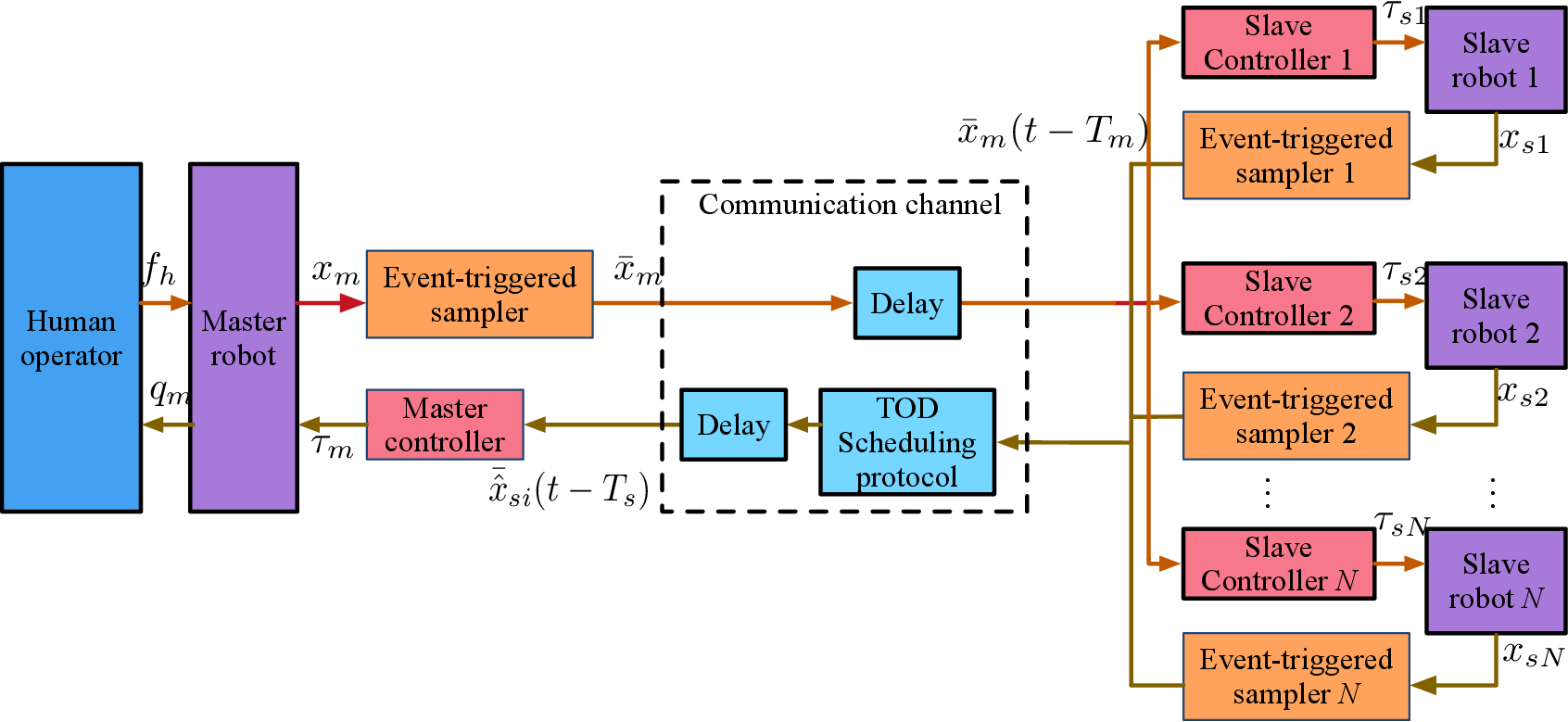}
	\caption{The  control framework under the event-triggered communication for the SMMS teleoperation system with the TOD scheduling protocol. }
	\label{fig:control setup_v}
\end{figure}
The triggered sampling instants are represented by the time sequences $s_{k_m}^m, s_{k_{si}}^{si} (i=1, 2, ... N)$ for the master and the $i$th slave robots, respectively, where $k_z\in\mathbb{N}$ satisfy $s_{0}^z=0$, $s_{k_z}^z<s_{k_z+1}^z$,  $z\in\mathcal{M}$.  On the slave side, the multiple slaves can not transmit their data simultaneously at a time due to the communication bandwidth limitation. Hence, the slave robots collectively generate the sequence of time instants $s_0^s, s_1^s,..., s_{k}^s, ..., k\in\mathbb{N}$ at which the slaves' data are transmitted to the master according to the TOD scheduling protocol.

In the event-based communication, the sampling instants in the master and the slave sides are generated by an event-triggered mechanism. Hence, in this case, $\bar{x}_m$ is defined as follows:
\begin{align}
	\bar{x}_m = x_m(s_{k_m}^m), t\in[s_{k_m}^m, s_{{k_m}+1}^m)\label{eq:bar_x_m_event}
\end{align}
Thus, the virtual observers (\ref{virtual_m_c})-(\ref{virtual_si_c}) with $\bar{x}_m$ given in (\ref{eq:bar_x_m_event}) are proposed in this event-triggered communication case. 

The proposed adaptive controllers under the event-triggered communication are presented by
\begin{align}\label{tol_event}
	\left\{
\begin{aligned}
&\tau_z\!=\!-\kappa_zr_z\!-\!Y_z\hat{\theta}_z,
\\
&\dot{\hat{\theta}}_z=\Gamma_z Y^{T}_z r_z, t\in[s_{k_z}^z, s_{k_z+1}^z), z\in\mathcal{M}
	\end{aligned}
	\right.
\end{align}
where $r_z$ and $Y_z$ are given in (\ref{s_variable}) and  (\ref{eq:Y_z1}), respectively, $\kappa_z>0$, $0< \Gamma_z \in \mathbb{R}^{n\times n}$. The event-triggered conditions are designed by 
\begin{align} \label{eq:event_iota2}
	s_{{k_z}+1}^z=\inf\left\{t|t>s_{k_z}^z,\iota_z^{'} \ge 0\right\}
\end{align}
where 
\begin{equation}\label{eq:triggering condition}
	\begin{aligned}
		\iota_z^{'}\!=\!|\delta_z|^2\! -\! \frac{2c_z(1\!-\!p_z)}{\beta_z}|\dot{x}_z|^2\!-\!\epsilon_z^{'} e^{\!-\!\nu_z^{'} t}, z\!\in\!\mathcal{M}\\
	\end{aligned}
\end{equation}
are the event-triggered functions, $c_z, \epsilon_z^{'}, \nu_z^{'} > 0, z \in \mathcal{M}$, and 
\begin{align}\delta_m\triangleq x_m-\bar{x}_m, 
	 \delta_{si}\triangleq x_{si}-\bar{x}_{si} \label{eq:delta}
\end{align}
with $\bar{x}_m, \bar{x}_{si}$ given in (\ref{eq:bar_x_m_event}) and (\ref{eq:bar_x_si}), respectively. 
Then, by substituting (\ref{virtual_m_c}), (\ref{eq:reconstructed_hat_xsi}), (\ref{virtual_si_c}), and (\ref{eq:bar_x_m_event}) into (\ref{eq:master})-(\ref{eq:slave}), we obtain the closed-loop teleoperation system model
\begin{equation}\label{closed-loop_v}
\left\{
\begin{aligned} 
f_m=&M_m(q_m) \dot{r}_m+C_m(q_m, \dot{q}_m) r_m+\kappa_m r_m
-Y_m \tilde{\theta}_m\\ 
f_{si}=&M_{si}(q_{si}) \dot{r}_{si}+C_{si}(q_{si}, \dot{q}_{si}) r_{si}+\kappa_{si} r_{si}-Y_{si} \tilde{\theta}_{si}\\
\ddot{x}_m=&-\alpha_m\dot{x}_m+k_me_m-\beta_m\left(x_m-\frac{1}{N}\sum_{i=1}^{N}x_{si}\right)\\
&-\frac{\beta_m}{N}\sum_{i=1}^{N}\!\left(L_{si}^{'}\!+\!\delta_{si}(t\!-\!T_s)\right)\!
+\frac{\beta_m}{N}\sum_{i\!=\!1,i\!\neq\! i_k^*}^N\eta_{i}(t\!-\!T_s)\\
\ddot{x}_{si}=&-\alpha_{si}\dot{x}_{si}+k_{si}e_{si}-\beta_{si}\left(\check{x}_{si}-x_m\right)\\
&-\beta_{si}\left(L_m^{'}+\delta_m(t-T_m)\right),\\
\dot{\eta}_{i}=&0, t \in [s_{k}^s, s_{{k}+1}^s), i= 1, 2,..., N, k \in \mathbb{N}
\end{aligned}\right.
\end{equation}
with the delayed reset system (\ref{eq:event_c_reset}), 
and  $
L_m^{'}\!=\!\int_{t\!-\!T_m}^t\dot{x}_m(v)dv$, 
$L_{si}^{'}\!=\!\int_{t\!-\!T_s}^t\dot{x}_{si}(v)dv
$.



The following candidate Lyapunov-Krasovskii functional for the closed-loop system (\ref{closed-loop_v}) is proposed:
\begin{equation}\label{V_v}
V^{'}=V_1+V_2^{'}+V_T, t \in [s_{k}^s,s_{{k}+1}^s)
\end{equation} 
with $V_1, V_T$ given in (\ref{V_v_c}),

\begin{align}
V_2^{'}=&N\int_{-d_m}^{0}\int_{t+v}^{t}\dot{x}_m^T(\sigma)R_m\dot{x}_m(\sigma)\,d\sigma\,dv\nonumber\\
&+\sum_{i=1}^{N} \int_{-d_s}^{0}\int_{t+v}^{t}\dot{x}_{si}^T(\sigma)R_{si}\dot{x}_{si}(\sigma)\,d\sigma\,dv\nonumber\\ &+\frac{N\beta_m}{2(1-p_m)}\int_{t-T_m}^{t}|\delta_m(\sigma)|^2\,d\sigma\nonumber\\
&+\sum_{i=1}^{N}\frac{\beta_m}{2(1-p_s)}\int_{t-T_s}^{t}|\delta_{si}(\sigma)|^2\,d\sigma 
\end{align}
where $R_m, R_{si}\in\mathbb{R}^{n\times n}$ are positive definite matrices. 

Similar to  Lemma \ref{LMIs_v}, it is easy to show that the Lyapunov functional $V^{'}$ is positive and does not grow at  $\black{s_{k+1}^s, k\in\mathbb{N}}$ if the LMIs (\ref{LMI_v}) hold with some $0<Q_i, U_i, P_i\in\mathbb{R}^{n\times n}$.


By now, we derive the following stability result.
\begin{Theorem}\label{thm:Theorem_comm}
	Consider the closed-loop teleoperation system (\ref{closed-loop_v}) with the event-triggered  mechanism (\ref{eq:event_iota2}) under the TOD scheduling protocol (\ref{eq:weighting_matrices}). If there exist positive-definite matrices $R_{m}, R_{si}, P_{i}, U_{i}, Z_{i}$ such that  the LMIs (\ref{LMI_v}) and the following LMIs are satisfied:
	
	\begin{align} \label{lmi_e}
			&\Pi_i
			\!=\!\begin{bmatrix}
				{\Pi_{i}^{11}} & {0}  & {0} &\Pi_{i}^{14}& {0} & {\Pi_i^{16}} \\
				{*} & {\Pi_{i}^{22}}  & \Pi_i^{23} & {0} & {0} & {0} \\ 
				{*} & {*} &{\Pi_{i}^{33}} & {0} & {0} & {0} \\ 
				{*} & {*} & {*} & {\Pi_{i}^{44}} & {0} & {0} \\ 
				{*} & {*} & {*} & {*} & {\Pi_i^{55}} & {0} \\  
				{*} & {*}  & {*} & {*} & {*} & {\Pi_i^{66}}
			\end{bmatrix}
			<0\\&	i  = 1, 2, \dots, N\nonumber
		\end{align}
		where
		\begin{align*}
					\Pi_{i}^{11}&=-\alpha_mI+\frac{\beta_m}{2}I+d_mR_m+c_mI,  \Pi_i^{14}\!=\!-\!\frac{\beta_m\!d_s}{2}I\\
				\Pi_{i}^{22}&=-\frac{\beta_m\alpha_{si}}{\beta_{si}}I+\frac{\beta_m}{2}I+d_sR_{si}+c_{si}I+hP_i\\
				\Pi_i^{23}&\!=\!-\!\frac{\beta_m\!d_m}{2}I, \Pi_{i}^{33}\!=\!-\!d_mR_m, \Pi_{i}^{44}\!=\!-\!d_sR_{si} \\
				\Pi_i^{16}&=\operatorname{row}_{j=1, \ldots, N}\left\{\frac{\beta_mh}{2}I, j \neq i\right\}\\
				\Pi_i^{55}&=-\operatorname{diag}_{j=1, \ldots, N}\left\{hU_{j}-Z_j, j \neq i\right\}\\
				\Pi_i^{66}&=- \operatorname{diag}_{j=1, \ldots, N}\left\{(1-p_s)Z_j, j \neq i\right\} 
	\end{align*}
   then the following claims hold:
	\begin{enumerate}
		\item[(1)] if the considered teleoperation system is in free motion, then all the signals are bounded, the system (\ref{closed-loop_v}) is asymptotically stable and $\lim_{t \rightarrow \infty}\left(x_m-\check{x}_{si}\right)=\lim_{t \rightarrow \infty}\dot{x}_m=\lim_{t \rightarrow \infty}\dot{x}_{si}=\lim_{t \rightarrow \infty}\left(q_m-\check{q}_{si}\right)=\lim_{t \rightarrow \infty}\dot{q}_m=\lim_{t \rightarrow \infty}\dot{q}_{si}=0$, $i=1,2,...,N$, $\lim_{t \rightarrow \infty}e_m =\lim_{t \rightarrow \infty}e_{si}=0 $;
		\item[(2)] if $f_m, f_{si} \in \mathcal{L}_{2} \cap \mathcal{L}_{\infty} $, then all the signals are bounded, and $\lim_{t \rightarrow \infty}\left(x_m-\check{x}_{si}\right)=\lim_{t \rightarrow \infty}\dot{x}_m=\lim_{t \rightarrow \infty}\dot{x}_{si}=\lim_{t \rightarrow \infty}\left(q_m-\check{q}_{si}\right)=\lim_{t \rightarrow \infty}\dot{q}_m=\lim_{t \rightarrow \infty}\dot{q}_{si}=\lim_{t \rightarrow \infty}e_m =\lim_{t \rightarrow \infty}e_{si}=0$;
		\item[(3)] the proposed event-triggered mechanism  excludes Zeno behavior.
\end{enumerate}
\end{Theorem}

\begin{proof}
	
	
	We first consider the case that the master and the slave manipulators are in free motion. The derivative of $V_1$ along with the trajectory of the system (\ref{closed-loop_v}) is given by

	\begin{align}\label{eq:dot_V_1}
	\dot{V}_1
	=&-\frac{N\kappa_m}{2\lambda_m}\dot{q}_m^T\dot{q}_m-\frac{N\kappa_m\lambda_m}{2}e_m^Te_m-N\alpha_m\dot{x}_m^T\dot{x}_m \nonumber\\
		&-\beta_m\!\sum_{i=1}^{N}\!\left(\!\dot{x}_{si}^T\!\left(L_m^{'}\!+\!\delta_m(t\!-\!T_m)\!\right)\!
	\!+\!\dot{x}_m^T\!\left(\!L_{si}^{'}\!+\!\delta_{si}(t\!-\!T_s)\!\right) \!\right)\nonumber\\
	&+\sum_{i=1}^{N}\bigg(-\frac{\beta_m\kappa_{si}}{2\beta_{si}\lambda_{si}}\dot{q}_{si}^T\dot{q}_{si}-\frac{\beta_m\kappa_{si}\lambda_{si}}{2\beta_{si}}e_{si}^Te_{si} \nonumber\\
	&-\frac{\beta_m\alpha_{si}}{\beta_{si}}\dot{x}_{si}^T\dot{x}_{si}\bigg)+\beta_m\dot{x}_m^T\sum_{i=1,i\neq i_k^*}^N\eta_{i}(t-T_s) 
\end{align}
	
	Calculating the time derivative of $V_2^{'}$ yields

	\begin{align} 
	\dot{V}_2^{'} \!\le \! &Nd_m\dot{x}_m^T  \!R_m \! \dot{x}_m \!- \!\frac{N}{d_m}L_m^{'T}R_mL_m^{'}  \!- \!\frac{N\beta_m}{2}|\delta_m(t \!- \!T_m)|^2  \nonumber\\
	&+\sum_{i=1}^{N} \!\left( \!d_s\dot{x}_{si}^T \! R_{si} \!\dot{x}_{si} \!- \!\frac{1}{d_s}L_{si}^{'T}R_{si}L_{si}^{'} \!- \!\frac{\beta_m}{2}|\delta_{si}(t \!- \!T_s)|^2  \!\right)  \nonumber\\
	&+\frac{N\beta_m}{2(1-p_m)}|\delta_m|^2+\sum_{i=1}^N\frac{\beta_m}{2(1-p_s)}|\delta_{si}|^2 
\end{align}

From the event-triggered conditions (\ref{eq:event_iota2}), one has $|\delta_m|^2< \frac{2c_m(1-p_m)}{\beta_m}|\dot{x}_m|^2+\epsilon_m^{'} e^{-\nu_m^{'} t}$ for $t \in (s_{k_m}^m, s_{{k_m}+1}^m)$ and $|\delta_{si}|^2< \frac{2c_{si}(1-p_s)}{\beta_m}|\dot{x}_{si}|^2+\epsilon_{si}^{'} e^{-\nu_{si}^{'} t}$ for $t \in (s_{k_{si}}^{si}, s_{{k_{si}}+1}^{si})$. \black{Denote 				$\phi_{m}^{'}\triangleq\frac{\kappa_m}{2\lambda_m}, \phi_{si}^{'}\triangleq\frac{\beta_m\kappa_{si}}{2\beta_{si}\lambda_{si}},							\psi_{m}^{'}\triangleq\frac{\lambda_m\kappa_m}{2}, \psi_{si}^{'}\triangleq\frac{\beta_m\lambda_{si}\kappa_{si}}{2\beta_{si}}$. } Therefore, one has  for $t \in [s_{k}^s, s_{{k}+1}^s)$,

\begin{align} 
	\dot{V}^{'} \!\le\!& \sum_{i=1}^{N}\!\Bigg( \!-\!\phi_{m}^{'}\dot{q}_m^T\dot{q}_m\!-\!\phi_{si}^{'}\dot{q}_{si}\dot{q}_{si}\!-\!\psi_{m}^{'}e_m^Te_m-\psi_{si}^{'}e_{si}^Te_{si}
	\nonumber\\
	&+\chi_{i}^T\underline{\Xi}^{'}\Pi_i\underline{\Xi}^{'}\chi_i
	+\frac{\beta_m}{2}\left(\frac{\epsilon_m^{'}e^{-\nu_m^{'} t}}{1-p_m}\!+\!\frac{\epsilon_{si}^{'}e^{-\nu_{si}^{'} t}}{1-p_s}\right)\Bigg)
\end{align}
where  $\underline{\Xi}^{'}=\text{diag}\left\{1,1,1/d_m,1/d_s,1/h,1/h \right\}\black{\otimes I}$.

%
	
	Through an argument similar to the proof of Theorem 1, one can easily reach that Claim (1) and Claim (2) hold by using  (\ref{lmi_e}).

	Finally, following the same line of reasoning as that of Claim 3 in Theorem \ref{thm:Theorem_c}, it can be shown that Claim 3 in Theorem \ref{thm:Theorem_comm} is true.  This completes the proof.
	
	
\end{proof}

The proposed control scheme for the SMMS teleoperation system (\ref{eq:master})-(\ref{eq:slave}) under an event-triggered communication and the TOD scheduling protocol is summarized in Algorithm~\ref{alg:comm}. 

\begin{algorithm}
\caption{Adaptive Control with event-triggered communication and the TOD scheduling protocol}\label{alg:comm}
	\hspace*{0.00in}{\bf Input:} Measurements $q_z, \dot{q}_z, x_z, \dot{x}_z$ from the controlled manipulator $z$, $z\in\mathcal{M}$; 
	virtual systems (\ref{virtual_m_c})  and (\ref{virtual_si_c}) with $\bar{x}_m$ given in (\ref{eq:bar_x_m_event}) ;\\
	\hspace*{0.00in}{\bf Output: }Event time $s_{k_z}^z$, $k_z= 1, 2, \cdots$, and control output $\tau_z(t)$, $z\in\mathcal{M}$;
	\begin{algorithmic}
		\FORALL{$z\in\mathcal{M}$}
		\WHILE{$t<s_{k_z}^z$}
		\STATE measure $q_z, \dot{q}_z$;
		\STATE receive $x_j$ from the remote side ($j=m$ if $z = {si}$ and $j={si}$ if $z=m$, $i= 1, \cdots, N$);
		\STATE compute  $ x_z, \dot{x}_z$ by evolving the virtual systems (\ref{virtual_m_c})  if $z = m$ and (\ref{virtual_si_c})  if $z=si$ with $\bar{x}_m$ given in (\ref{eq:bar_x_m_event}); 
		\STATE compute $e_z, \dot{e}_z$ from (\ref{errors_v}); compute $r_z$ from (\ref{s_variable}); compute $Y_z$ from (\ref{eq:Y_z1});
		\STATE compute $\hat{\theta}_z$ by integrating $\dot{\hat{\theta}}_z$ from $\hat{\theta}_z(0)$ using (\ref{tol_event});
		\STATE compute $\iota_z^{'}$ from (\ref{eq:triggering condition})-(\ref{eq:delta});
		\IF{$s_{k_z+1}^z$ is found with (\ref{eq:event_iota2})}
		\STATE  $k_z \leftarrow k_z+1 $; 
		\IF{$z=m$}
		\STATE transmit $x_m(s_{k_m}^m)$ to the slave side;
		\ELSE 
		\STATE pick $i_k^*$ by (\ref{eq:weighting_matrices}) and transmit $x_{si}(s_{k_{si}}^{si})$ with $i = i_k^*$ to the master side;  set $s_k^s = t$, $k\leftarrow k+1$; 
		\ENDIF
		\ENDIF
       \STATE compute $\tau_z(t)$ from (\ref{tol_event});
		\ENDWHILE		
		\ENDFOR
	\end{algorithmic}
\end{algorithm}

%

\begin{remark}
Compared with \cite{li2019bilateral} in which the P+d  control of SMMS teleoperation systems with the RR and the TOD protocols was addressed, in this section we propose the adaptive controllers for SMMS teleoperation systems by simultaneously considering the TOD protocol and the event-triggered mechanisms. Thus, the communication bandwidth usage can be further improved and data collisions can be avoided. In addition,  the proposed controllers do not require relative velocity measurements, which eliminates the need for real-time transmission of velocity information across the network. This makes their application more practical, as the transmission of velocities typically demands high communication bandwidth and may introduce noises.

\end{remark}

\begin{remark}\label{rek:parameter_choose}
For both control schemes (\ref{tol_c}) and (\ref{tol_event}), the selection of control parameters should consider the following guidelines based on the analysis of tracking performance, stability, and robustness against communication constraints:
(i) The feedback gain $\kappa_z$ for $z \in \mathcal{M}$ should be chosen in accordance with the saturation limits and bandwidth of the actuators, ensuring that the control inputs do not exceed actuator capacity;
(ii) Larger values of $\lambda_z$ improve tracking performance by increasing the influence of the tracking error, while larger values of $\beta_z/\alpha_z$ enhance convergence speed. However, excessively large values of $\beta_z/\alpha_z$ can negatively impact stability, especially under long sampling intervals or in the presence of variable network delays;
(iii) The proposed control schemes exhibit improved delay robustness with a higher $\alpha_z/\beta_z$ ratio (or equivalently a smaller $\beta_z/\alpha_z$), which allows for larger sampling intervals while maintaining system stability and reducing sensitivity to communication constraints.
\end{remark}

\begin{remark}
The main results are applicable to the case where the systems are subject to disturbances,   meaning that exists a disturbance $D_i\in\mathcal{L}_2\cap\mathcal{L}_\infty$ for each manipulator such that the dynamic model of the manipulators in the SMMS teleoperation system is described as:
		\begin{align}
			M_i(q_i)\ddot{q}_i+C_i(q_i, \dot{q}_i)\dot{q}_i +G_i(q_i) +D_i = \tau_i+f_i
			\end{align}
	It can be easily deduced that if  $D_i\in\mathcal{L}_2\cap\mathcal{L}_\infty$, then the considered teleoperation system with the external force $f_i\in\mathcal{L}_2\cap\mathcal{L}_\infty$ is stable with that all the signals are bounded, and $\lim_{t \rightarrow \infty}\left(x_m-\check{x}_{si}\right)=\lim_{t \rightarrow \infty}\dot{x}_m=\lim_{t \rightarrow \infty}\dot{x}_{si}=\lim_{t \rightarrow \infty}\left(q_m-\check{q}_{si}\right)=\lim_{t \rightarrow \infty}\dot{q}_m=\lim_{t \rightarrow \infty}\dot{q}_{si}=\lim_{t \rightarrow \infty}e_m =\lim_{t \rightarrow \infty}e_{si}=0$. This conclusion can be derived by applying Claims 2 of Theorems 1 and 2, since the system can be reformulated as:
			\begin{align}
		M_i(q_i)\ddot{q}_i+C_i(q_i, \dot{q}_i)\dot{q}_i +G_i(q_i)= \tau_i+\bar{f}_i
	\end{align}
with $\bar{f}_i = f_i-D_i\in\mathcal{L}_2\cap\mathcal{L}_\infty$
\end{remark}

\section{Discussion}\label{sec:discussion}

In this paper, we consider a class of unknown teleoperation systems involving a single master and multiple slaves, connected over a wide-area communication network. To mitigate these issues, the TOD scheduling protocol is used to determine the sequence of data transmission from the slave side. Additionally, to further reduce network load and improve computational efficiency, event-triggered communication and event-triggered control updates are implemented, respectively.

Establishing a model that simultaneously accounts for event-triggered mechanisms and the TOD scheduling protocol, while developing adaptive control strategies to ensure the desired performance of SMMS teleoperation systems, is a challenging task. The use of the TOD protocol complicates controller design due to the change in the order of data transmission, leading to discontinuities. Specifically, differentiation of the information received from the remote side is typically required in adaptive control design, but the TOD scheduling mechanism conflicts with this process. Under the TOD protocol, data transmission from the slaves to the master is discontinuous. In (\ref{eq:Y_z1}), the derivative of $e_z$ is needed to formulate $Y_z$ and $\dot{r}_z$, which is continuous in systems without the TOD scheduling protocol and the event-triggered mechanisms, but is discontinuous in our case. In our approach, $e_z$ is redefined as the position error between the manipulator $z$ and the virtual system $z$, making it continuous and ensuring that $\dot{e}_z$ always exists. This facilitates the stability analysis of the closed-loop system. The introduction of virtual observers also eliminates the need to transmit velocities over the network and simplifies the event-triggered conditions.
Additionally, the system models, such as (\ref{closed-loop_v_c}) and (\ref{closed-loop_v}), are hybrid time-delay systems with reset conditions (\ref{eq:event_c_reset}), making the stability analysis more complex. To address this, we propose non-continuous Lyapunov-Krasovskii functions (\ref{V_v_c}) and (\ref{V_v}) to handle these challenges. 

\black{It should also be noted that the event-triggering conditions in (\ref{eq:event1}) and (\ref{eq:event_iota2})  are static rather than dynamic \cite{Xu2025TAC}. While dynamic event-triggering conditions, which include a dynamic term related to the relative threshold, could improve transmission efficiency, they also introduce additional computational demands and increase the complexity of control updates, as the dynamic term requires tuning via an updating algorithm. Given that the proposed adaptive control algorithms already include virtual systems, adding to the computational load, we chose not to use dynamic triggering conditions to conserve computational resources. 
Additionally, as highlighted in Remark~\ref{rek:mininum_IET}, although the proposed event-triggered mechanisms avoid Zeno behavior, the uniform lower bound of the inter-event intervals cannot be guaranteed, which is an issue that should be addressed in future work.}

As for large-scale teleoperation scenarios, the proposed controllers (\ref{tol_c}) and (\ref{tol_event}) adopt a distributed architecture, which enhances their practicality. To cope with the increased computational complexity as the number of slave manipulators grows greatly, the controllers can be deployed in a fully distributed manner, with appropriately tuned parameters to reduce communication and computation loads. Additionally, hardware acceleration techniques can be utilized to further improve efficiency, thereby supporting scalability and real-time performance.


\section{Conclusion}\label{sec:Conclusion}

In this paper, the event-triggered control and communication mechanisms  for a class of SMMS teleoperation systems with TOD scheduling protocol have been investigated, respectively.  To avoid using relative velocities, the virtual system for each manipulator has been constructed.  Two adaptive control algorithms have been proposed with the presented event-triggered mechanisms. By the Lyapunov–Krasovskii functionals, the stability criteria  have been obtained for the closed-loop system under the two proposed control schemes, respectively. The sufficient conditions related to the adaptive controller gains, the virtual systems' parameters, the upper bound of communication delays, and the maximum allowable sampling interval have been provided. Finally, the effectiveness of the proposed algorithms has been verified by the experimental results. \black{Future works will include the safety-guaranteed control design of teleoperation systems interacting with complex environments under communication constraints.}


\bibliographystyle{IEEEtran_doi}
\bibliography{document_R1}

\def\@IEEEBIOskipN{1\baselineskip}

\begin{IEEEbiography}[{\includegraphics[width=1in,height=1.25in,clip,keepaspectratio]{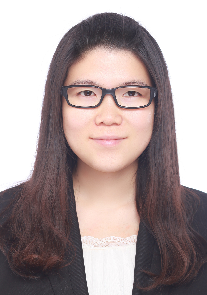}}]{Yuling Li}(M'18)
	received the B.Eng. and Ph.D. degrees  from University
	of Science and Technology Beijing (USTB),
	Beijing, China, in 2009 and 2016, respectively.
	She was a visiting scholar at the Department of Automatic Control, Lund University,
	Sweden, from Sep. 2012-Sep.
	2014, and at Nanyang Technological University, Singapore, from Jan. 2024 to Jul. 2024.  In January 2016 she joined the School of Automation and Electrical Engineering, USTB, Beijing, China, and now is an Associate Professor. Her research interests include robotics, networked systems and nonlinear system control. 
\end{IEEEbiography}

\begin{IEEEbiography}[{\includegraphics[width=1in,height=1.25in,clip,keepaspectratio]{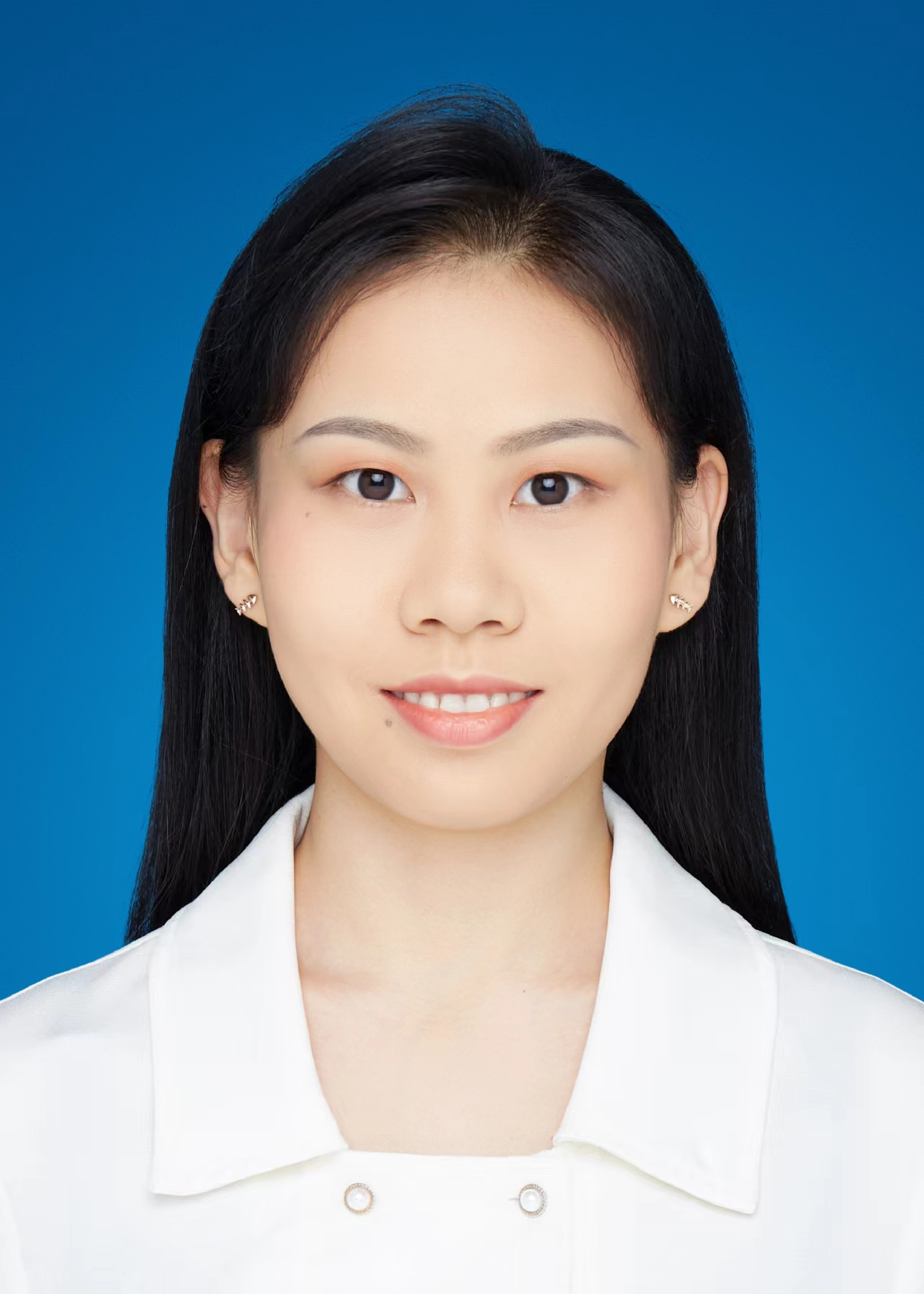}}]{Chenxi Li}
	received the Bachelor's degree from Huaqiao University, Xiamen, China, in 2020, and the Master's Degree in Control Science and Engineering from University of Science and Technology Beijing, Beijing, China, in 2023. Her current research interests include robotics and teleoperation systems.
\end{IEEEbiography}

\begin{IEEEbiography}[{\includegraphics[width=1in,height=1.25in,clip,keepaspectratio]{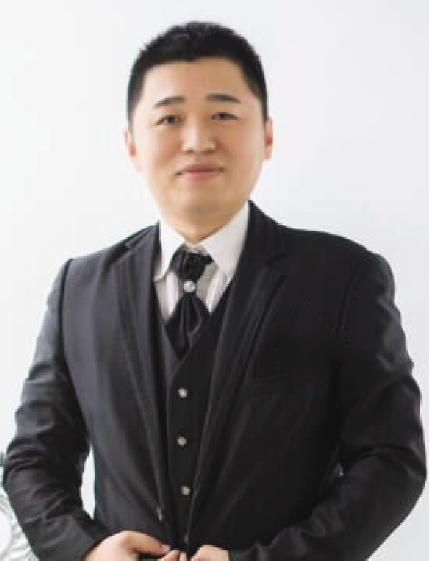}}]{Kun Liu} (M'16-SM'21) received the Ph.D. degree in electrical engineering and systems from Tel Aviv University, Tel Aviv-Yafo, Israel, in 2012. From 2013 to 2015, he was a Postdoctoral Researcher with the ACCESS Linnaeus Centre, KTH Royal Institute of Technology, Stockholm, Sweden. In 2015, he joined the School of Automation, Beijing Institute of Technology, Beijing, China. His current research interests include networked control, game-theoretic control, and security and privacy of cyber-physical systems, with applications in autonomous systems. 
\end{IEEEbiography}
\begin{IEEEbiography}[{\includegraphics[width=1in,height=1.25in,clip,keepaspectratio]{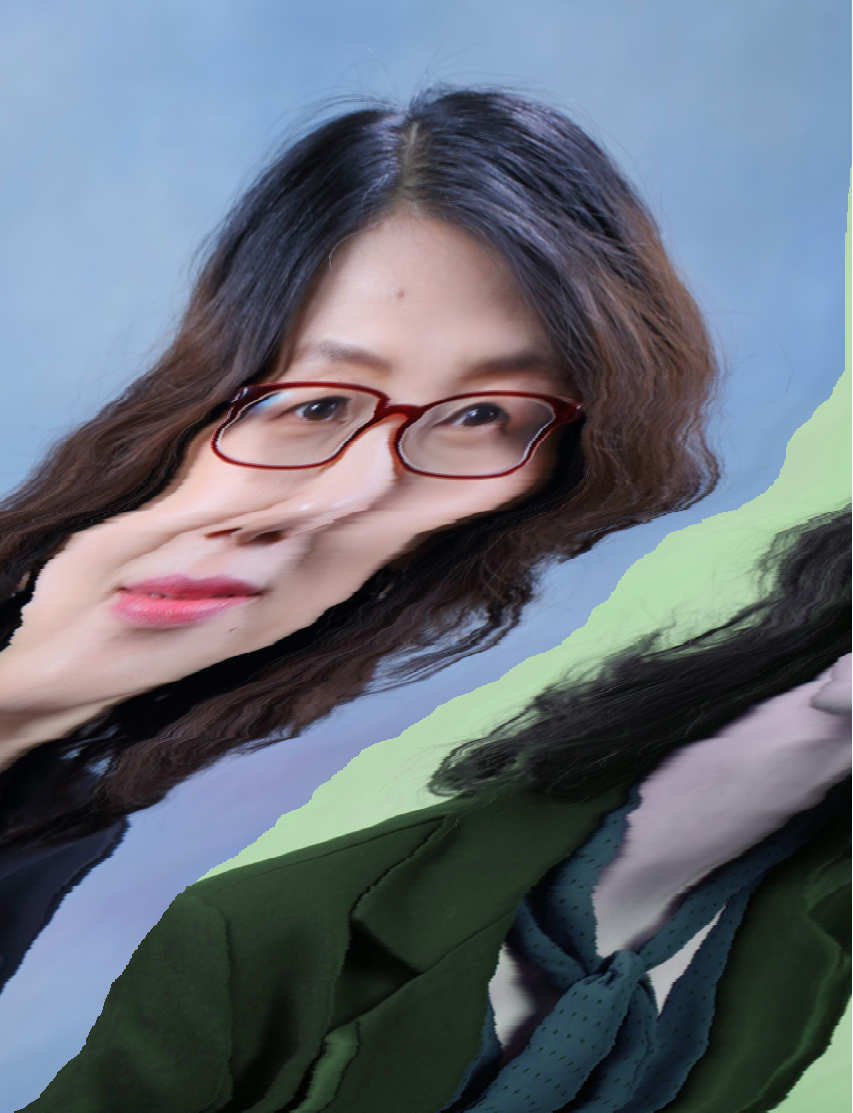}}]{Jie Dong}
received the B.E., M.E., and Ph.D. degrees from  University of Science and Technology Beijing, Beijing, China, in 1995, 1997, and 2007, respectively. From July 2004 to December 2004, she visited University of Manchester, Manchester, U.K., as a Visiting Scholar. She is currently a Professor with the School of Automation and Electrical Engineering, University of Science and Technology Beijing. Her research interests cover intelligent control theory and applications, process monitoring and fault diagnosis, and complex system modeling and control.
\end{IEEEbiography}
\begin{IEEEbiography}[{\includegraphics[width=1in,height=1.25in,clip,keepaspectratio]{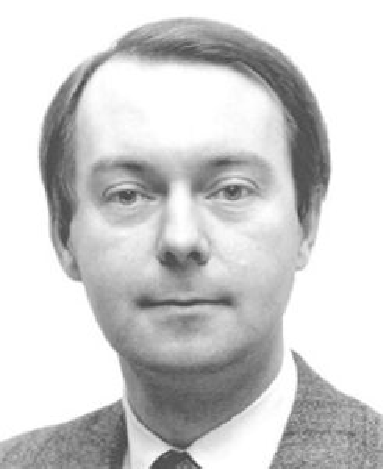}}]{Rolf Johansson}(F'12)
	received the Master-of-Science degree in Technical Physics in 1977, the Bachelor-of-Medicine degree in 1980, the doctorate in control theory in 1983, was appointed Docent in 1985, and received the Doctor-of-Medicine degree (M.D.) in 1986, all from Lund University, Lund, Sweden.
	He is IEEE Fellow; Fellow of the Swedish Society of Medicine; Member of B\'ar\'any  Society, and Fellow of the Royal Physiographic Society, Section of Medicine. He is Assoc. Editor of Int. J. Adaptive Control and Signal Processing and Editor of Mathematical Biosciences.\end{IEEEbiography}
\vfill

\end{document}